\documentclass{article}
\usepackage{amsmath,amssymb,amsfonts,amsthm}
\usepackage{a4wide}
\usepackage{color}
\usepackage{float}
\usepackage{comment}
\usepackage{authblk}
\usepackage{graphicx} 
\usepackage{cases}

\usepackage{thmtools}
\usepackage{thm-restate}

\def\R{\mathbb{R}}
\def\E{\mathbb{E}}
\def\P{\mathbf{P}}
\def\H{\mathbf{H}}

\newcommand{\RV}{\textcolor{black}}

 \newcommand{\red}{\textcolor{black}}

\newtheorem{remark}{Remark}
\newtheorem{lemma}{Lemma}

\newtheorem{proposition}{Proposition}
\newtheorem*{proposition*}{Proposition}

\title{Model Reduction of Multivariate Geometric Brownian Motions and Localization in a Two-State Quantum System}
\author[$\dagger$]{C. Chen}
\author[$\star$]{M. Colangeli\thanks{Corresponding author: \texttt{matteo.colangeli1@univaq.it}}}
\author[$\dagger$]{M. H. Duong}
\author[$\star$]{M. Serva}

\affil[$\dagger$]{School of Mathematics, University of Birmingham, UK. }
\affil[$\star$]{Department of Information Engineering, Computer Science and Mathematics, University of L’Aquila, L’Aquila, Italy.}

\date{\today}

\begin{document}

\maketitle
\begin{abstract}
We develop a systematic framework for the model reduction of multivariate geometric Brownian motions (GBMs), a fundamental class of stochastic processes with broad applications in mathematical finance, population biology, and statistical physics. Our approach leverages the interplay between the method of invariant manifolds and adiabatic elimination to derive closed-form reduced equations for the deterministic drift. An extended formulation of the fluctuation-dissipation theorem is subsequently employed to characterize the stochastic component of the reduced description. As a concrete application, we apply our reduction scheme to a GBM arising from a two-state quantum system, showing that the reduced dynamics accurately capture the localization properties of the original model while significantly simplifying the analysis.
\end{abstract}

\section{Introduction}
 
Complex systems are ubiquitous across the physical, biological, and social sciences, encompassing phenomena as diverse as turbulent fluid flows, neural networks, crowd dynamics, and opinion formation. Such systems are frequently modeled via coupled stochastic differential equations (SDEs), which characterize the temporal evolution of numerous interacting microscopic constituents (such as particles, spins, or agents) often subject to external environments or reservoirs. The stochastic terms in these equations account for the intrinsic uncertainty in the mechanisms driving the microscopic dynamics, which may stem from environmental fluctuations, measurement noise, or unresolved degrees of freedom \cite{bianconi2023complex}.
This modeling approach proves especially relevant in fields such as climate dynamics, for example in the seminal work of K. Hasselmann \cite{hasselmann1976stochastic} and the recent comprehensive review \cite{lucarini2023theoretical}, where the stochastic framework provides a basis for characterizing large-scale behaviors emerging from complex small-scale interactions.

Despite their descriptive accuracy, SDE-based models pose significant analytical and computational challenges, primarily due to the high dimensionality inherent in systems with a large number of interacting variables. This necessitates model reduction strategies that approximate the original system with a lower-dimensional surrogate, preserving the essential features of the dynamics while simplifying the analysis and reducing computational costs.

A key early impetus for model reduction came from statistical mechanics, most notably through N. N. Bogoliubov's foundational work in the kinetic theory of gases \cite{Bogol}, which sought to rigorously derive fluid dynamics from the Boltzmann equation. In a similar vein, the Chapman–Enskog expansion provided a systematic framework for constructing closed-form hydrodynamic equations and transport coefficients from kinetic theory \cite{chapman1990mathematical}. A further major development was the Mori–Zwanzig formalism \cite{mori1965transport, zwanzig1960ensemble}, which employs projection operators to reduce the Liouville equation to a generalized Langevin equation, effectively capturing the dynamics of a selected set of coarse-grained (or resolved) variables. These approaches laid the groundwork for modern reduction methods centered on memory kernels, fluctuation-dissipation relations, and coarse-graining techniques \cite{chorin2000optimal, majda2001mathematical, Vulp}.
Over the past two decades, the field of model reduction has matured into a rich and interdisciplinary area, extending far beyond its roots in statistical physics and finding applications in domains such as climate modeling \cite{hasselmann1976stochastic, lucarini2023theoretical}, systems biology \cite{snowden2017methodologies}, control theory \cite{antoulas2005approximation}, chemical kinetics \cite{lam1994dynamic, roberts2008model}, and fluid dynamics \cite{rowley2004model}. Comprehensive reviews and frameworks have further formalized the mathematical foundations and computational strategies for model reduction; see, for example, \cite{givon2004extracting, gorban2006model, benner2017model, kevrekidis2003equation, ott2005low}.

A widely adopted strategy in model reduction involves projecting high-dimensional Markovian dynamics onto a lower-dimensional manifold, typically defined in terms of a subset of ``slow'' or ``collective'' variables \cite{hartmann}. However, this dimensionality reduction often introduces memory effects, resulting in non-Markovian dynamics for the reduced model \cite{Zwanzig}.
The Markov property can be preserved in special cases, such as fast–slow systems, which are characterized by a clear time scale separation and exhibit a well-defined dynamical structure. On short timescales, fast variables rapidly relax toward a neighborhood of a lower-dimensional slow manifold. Subsequently, the system evolves along this manifold over much longer timescales until reaching a stationary state \cite{haken1983synergetics}. In the limit of infinite separation, one can rigorously derive reduced Markovian models for the slow variables alone \cite{Zwanzig, Ghil}, as systematically explored in works such as \cite{pavliotis2008multiscale, gorban2006model}. However, many real-world systems lack a clear time scale separation, prompting the development of alternative model reduction techniques suitable for such settings \cite{TonyRoberts}.
Several recent contributions have addressed this challenge. For instance, \cite{Wouters2019, Wouters2019b} propose reduction techniques based on the Edgeworth expansion, applicable to both deterministic and stochastic systems with moderate time scale separation. Meanwhile, \cite{Checkroun1, Checkroun2} leverage Ruelle–Pollicott resonances to develop reduced models for weakly separated regimes. Other studies, including \cite{Legoll2010, Lu2014, zhang2016effective, duong2017variational, Duong2018, Legoll2019, Lelievre2019, Hartmann2020, duong2025coarse}, focus on diffusion processes and employ conditional expectations to build reduced models, avoiding the need for explicit scale separation.

An especially elegant and geometrically motivated approach to model reduction is the method of the \emph{invariant manifold} (IM), which has seen renewed interest in recent years.  
\RV{Originally formulated within the framework of Kolmogorov–Arnold–Moser (KAM) theory for perturbed integrable Hamiltonian systems~\cite{Kol,Arn,Mos}, this method leverages the structure of a dynamical system to identify low-dimensional manifolds that remain approximately invariant under the flow.
The method proceeds by solving a system of invariance equations which, from a geometric standpoint, ensure that the vector field of the original system remains tangent to the manifold at every point. In the context of kinetic theory, the IM method was later adapted to derive closed-form hydrodynamic equations from the Boltzmann equation~\cite{chapman1990mathematical,Gor04,colan08,colan09}, enabling an exact resummation of the Chapman–Enskog expansion~\cite{GorKar,Kar02,CM22}. This approach effectively prevents the spurious divergences in dispersion relations that are typically encountered in high-order truncations~\cite{Bobylev,colan07}.
Recently, the method has been successfully applied across diverse disciplines, including the reduction of complex reaction networks in chemical kinetics~\cite{Gor2018}, climate modeling~\cite{Gin2014}, and the dimensional reduction of stochastic differential equations~\cite{CM22,CDM22,colangeli2025hybrid}.}
\RV {A recent advancement in this field} is the hybrid model reduction framework introduced in \cite{CM22, CDM22, CDM23, colangeli2025hybrid}, targeting linear SDEs with additive noise. The approach proceeds in two stages. First, the deterministic component of the dynamics is reduced using the invariant manifold method. Stochasticity is then incorporated into the contracted description through the fluctuation-dissipation \red{theorem, which links the diffusion matrix, the drift matrix, and the stationary \RV{second-moment} matrix \cite{Zwanzig,Vulp,Vulp09,Pavl}, so as to preserve the marginal stationary distribution of the full system \footnote{\red{Following Kubo \cite{Kubo66}, we refer here to the so called ``fluctuation-dissipation theorem of the \textit{second kind}'', which relates the friction coefficient in the generalized Langevin equation to the properties of the noise, see also \cite{Maes14}. By contrast, the ``fluctuation-dissipation theorem of the \textit{first kind}'' expresses the mobility coefficient in terms of the velocity correlation function and therefore constitutes a prototypical example of Linear Response Theory \cite{Kubo,Hang82,Col12}.}}}.
Notably, this procedure does not rely on the presence of a time scale separation, \red{that is, the original system is not required to exhibit a slow-fast structure}, which makes the method applicable to a broader class of dynamical systems.

In this paper, we extend the hybrid IM method to a new and widely applicable class of models, namely multivariate \emph{geometric Brownian motions} (GBMs). These are SDEs in which both the drift and diffusion terms scale linearly with the state, leading to multiplicative noise. GBMs play a foundational role in fields such as mathematical finance, population biology, and statistical physics \cite{mao2007stochastic, kloeden2013numerical}. However, their multiplicative noise structure precludes the direct use of the additive-noise reduction techniques developed in \cite{CM22, CDM22, CDM23}. Additionally, the conditional expectation method is not applicable to GBMs, as they generally do not admit an \RV{absolutely continuous (with respect to Lebesgue measure)} invariant measure, which is a key assumption in that approach.

Our main contribution in this paper is threefold:
(i) we extend the hybrid IM method to handle GBMs with multiplicative noise, thereby broadening its applicability to nonlinear stochastic systems of practical relevance;
(ii) we introduce a complementary model reduction strategy based on the analysis of higher-order ODEs governing the evolution of low-order statistical moments, offering analytical insight into the structure of the reduced dynamics;
(iii) we highlight the effectiveness of both reduction techniques on a GBM model derived from the study of localization phenomena in a two-state quantum system, originally proposed in~\cite{blanchard1994localization}. This application shows that our method successfully captures essential dynamical features of the model, such as localization behavior, in the strong-noise regime. \red{The fact that the contracted description retains the localization effects observed in the microscopic dynamics represents a significant advancement in the use of the IM method beyond its traditional domain of kinetic theory. In that context, the IM method is designed to ensure two fundamental properties of the reduced dynamics, namely the preservation of conservation laws and the preservation of dissipation \cite{GorKar,Stasenko24}. Our analysis demonstrates that, when the IM method is carefully combined with the fluctuation-dissipation theorem for GBMs, the guiding principles of the reduction scheme can be extended, allowing localization in quantum systems to naturally emerge as an additional physical observable that can be systematically preserved.}

Related works on model reduction for GBMs include \cite{benner2011lyapunov, damm2014balanced, benner2015model}, where the balanced truncation method is employed. These approaches also rely on  Lyapunov-type equations associated with GBMs (see e.g. Eq.~\eqref{eq: stationary covariance matrix} below).

The paper is organized as follows. In Section~\ref{sec: two-state system}, we revisit the two-state quantum system from Ref.~\cite{blanchard1994localization}, which provides the physical motivation for our reduction procedure. Notably, the authors in~\cite{blanchard1994localization} map the Schr\"odinger equation to a system of SDEs with purely geometric noise, naturally framing the model within the GBM formalism. Sections~\ref{sec: higher order ODEs} and~\ref{sec: IM method} then develop two reduced models for the quantum GBM: one based on higher-order ODEs, and the other on the IM method. Section~\ref{sec:reddyn} highlights the key properties of the resulting contracted description and provides the corresponding error estimates. Conclusions are drawn in Section~\ref{sec: conclusion}, while technical derivations and detailed computations are deferred to the Appendices. \RV{Specifically, Appendix~\ref{app:appA} addresses a two-variable formulation of the two-state quantum system, and Appendix~\ref{sec: GBMs} reviews relevant properties of multivariate GBMs, discussing model reduction strategies for cases where the reduced dynamics is multidimensional. Appendix~\ref{app:app2} provides the calculation of the eigenvalues of the matrix M (as defined in \eqref{eq: 4ODEs}) utilized in the preceding sections. Furthermore, Appendix~\ref{app:app3} derives the exact fourth-order linear ODE for the second moment of the reduced dynamics, while Appendix~\ref{app:app4} outlines an algebraic interpretation of the invariance equations. Finally, Appendix~\ref{app:app6} is devoted to the proofs of technical lemmas, and Appendix~\ref{app: L2norm} provides suitable error estimates.}

\section{A GBM from a two-state quantum system}
\label{sec: two-state system}

In this section, we revisit the two-state quantum system originally introduced in~\cite{blanchard1994localization}. This model exhibits particularly interesting physical behavior, characterized by partial transition inhibition and loss of periodicity in the strong noise regime. The Schr\"odinger equation for this system can be mapped to a set of SDEs equipped with purely geometric noise terms, making it an ideal test case for our proposed methodology.

\subsection{A two-state quantum system}
\label{sec: two state system}
Consider a system with two classical equilibrium configurations, the typical example being the double-well problem. A classical object with minimum energy is located in one of the two wells, while the ground state of a quantum object is symmetric and, therefore, delocalized. If one prepares a quantum state concentrated in one of the wells, for example by means of a wave function that is a superposition of the (symmetric) ground state and the first excited state (antisymmetric), the system will move coherently and periodically from one well to the other and, therefore, permanent localized states are not possible. The frequency will be $\omega= \frac{\Delta E}{\hbar}$ where $\Delta E$
is the difference between the energy of the first excited state and that of the ground state.

The simplest toy model that can illustrate this behavior can be constructed by means 
of the (spin) Hamiltonian 
\begin{equation}
\hat H= \alpha \hat \sigma_x,
 \label{h}
\end{equation}
where \RV{$\hat\sigma_x$ is the first of the three Pauli matrices}
\RV{
\[
\hat\sigma_x=\begin{pmatrix}
    0&1\\1&0
\end{pmatrix},\quad\hat\sigma_y=\begin{pmatrix}
    0&-i\\
    i&0
\end{pmatrix},\quad \hat\sigma_z=\begin{pmatrix}
    1&0\\
    0&-1
\end{pmatrix},
\]
}
which applies on a complex-valued,
two component, spinor 
\begin{equation}
\boldsymbol \psi =
 \begin{pmatrix}
\psi_l \\
\psi_r
\end{pmatrix}  .
 \label{hg}
\end{equation}
The Schr\"odinger equation can be written as
\begin{equation}
\frac{d \boldsymbol \psi}{dt} = - i\hat H \boldsymbol\psi,
\end{equation}
where, without loss of generality,  we have set $\hbar=1$.
It follows immediately
\begin{equation}
\frac{d\boldsymbol\psi }{dt} =-i\alpha \hat \sigma_x \boldsymbol\psi
\,\,\,\,\,\, \RV{\Longleftrightarrow}\,\,\,\,\,\,\,
\frac{d }{dt} 
 \begin{pmatrix}
\psi_l \\
\psi_r
\end{pmatrix}  
= - i\alpha 
 \begin{pmatrix}
0 &  1\\
1 &  0
\end{pmatrix} 
 \begin{pmatrix}
\psi_l \\
\psi_r 
\end{pmatrix},
 \label{hsx}
\end{equation}
where $|\psi_l(t)|^2$ is the probability that the system is found (by a measurement)
in the left well at time $t$ and  $|\psi_r(t)|^2$ is the probability that it is found in the right one
(the Schr\"odinger equation preserves normalization so that the scalar product 
$|\boldsymbol \psi(t)|^2 = |\psi_l(t)|^2+|\psi_r(t)|^2$  can be set to be
equal to 1 at all times).

The only two (normalized to 1) eigenstates of the Hamiltonian are 
\begin{equation}
\boldsymbol \psi_0 =
\frac 1 {\sqrt 2}
 \begin{pmatrix}
\,\, \,\, 1\\
-1
\end{pmatrix}  
\,\,\,\,\,\,\,\,\, {\rm and} \,\,\,\,\,\,\,\,\,
\boldsymbol \psi_1 =
\frac 1 {\sqrt 2}
 \begin{pmatrix}
1\\
1
\end{pmatrix} ,
 \label{eig}
\end{equation}
where $\boldsymbol \psi_0 $ is the ground state with energy $E_0=- \alpha$ 
and $\boldsymbol \psi_1 $
is the excited state with energy $E_1=\alpha$. Note that neither state is localized, on the 
contrary, the probability of finding the system in a given hole by measurement is half for 
both states. If the initial state is prepared so that it is localized at time t = 0, for example in the left well
which means $\boldsymbol \psi (0) =(\boldsymbol \psi_0 +\boldsymbol \psi_1) /\sqrt 2
= (1,0)$, its Schr\"odinger evolution at time $t$ is
\begin{equation}
\boldsymbol \psi (t) = \frac 1 {\sqrt 2} \left( \boldsymbol \psi_0 e^{-i \alpha t} +
\boldsymbol \psi_1 e^{i \alpha t} \right) =
 \begin{pmatrix}
\cos ({\alpha t} )\\
\sin({\alpha t} )
\end{pmatrix} .
 \label{sup}
\end{equation}
One has, therefore, $|\psi_l(t)|^2= |\cos ({\alpha t} )|^2$ and  $|\psi_r(t)|^2= |\sin({\alpha t}) |^2$ 
which are periodic with frequency 
$\omega=2\alpha=  E_1 \! -\!  E_0  = \Delta E/\hbar$. This means that the system 
is perfectly localized in the left well at any time $t=\pi n/  \alpha$ and 
it is perfectly localized in the right well at any time $t=\pi (n+1/2)/ \alpha$. 

We pose the following question: why is such coherent behavior not observed in mesoscopic systems? The standard answer is that no system is ever truly isolated, and once interactions with the environment are taken into account, the phenomenology changes significantly. The simplest way to introduce interactions with the environment is to assume that they are numerous but individually weak, possibly arising from various sources. Starting from the toy model described above, a natural extension is to include an additional random interaction term to account for this environmental coupling. The Hamiltonian (\ref{h}) is consequently replaced by
\begin{equation}
\hat H= \alpha \hat \sigma_x +  \beta \eta(t)  \hat \sigma_z ,
 \label{h2}
\end{equation}
where \RV{$\hat\sigma_z$} is the Pauli matrix and $\eta(t)$ is a white noise. 
The Schr\"odinger equation becomes
\begin{equation}
 d\boldsymbol\psi  =-i\alpha \hat \sigma_x \boldsymbol\psi dt
-i \beta  \hat \sigma_z  \boldsymbol\psi  \, dW_t- \frac 1 2 \beta^2 \boldsymbol\psi dt
  \label{hsxb}
\end{equation}
which, in a more explicit form, can be rewritten as
\begin{equation}
d
 \begin{pmatrix}
\psi_l \\
\psi_r
\end{pmatrix}  
= - i\alpha 
 \begin{pmatrix}
0 &  1\\
1 &  0
\end{pmatrix} 
 \begin{pmatrix}
\psi_l \\
\psi_r 
\end{pmatrix}
dt 
- i \beta 
 \begin{pmatrix}
1 &  0\\
0 &  -1
\end{pmatrix} 
 \begin{pmatrix}
\psi_l \\
\psi_r 
\end{pmatrix}
dW_t
- \frac 1 2 \beta^2
 \!  \begin{pmatrix}
\psi_l \\
\psi_r 
\end{pmatrix}
dt,
  \label{hsx2}
\end{equation}
where  $W_t= \int_0^t \eta(s) ds$ is the standard one-dimensional Wiener process and 
where the last term appears
because we use  It${\hat {\rm o}}$'s notation (if Stratonovich notation were used instead, 
this would be absent).

It can be directly verified that the Schr\"odinger equation preserves normalization. In fact,
having defined the row spinor $\boldsymbol \psi^+\! =(\psi_l^*,  \, \psi_r^*)$,
one has 
\begin{equation}
 d\boldsymbol\psi^+ =i\alpha \boldsymbol \psi ^+\hat \sigma_x \, dt
+ i \beta  \boldsymbol \psi^+  \hat \sigma_z  dW_t - \frac 1 2 \beta^2 \boldsymbol \psi^+dt,
  \label{hsxbbis}
\end{equation}
then, the scalar product  $\boldsymbol \psi^+\!(t)  \boldsymbol \psi (t)
= |\boldsymbol \psi(t)|^2 = |\psi_l(t)|^2+|\psi_r(t)|^2$ is constant,
as it can be verified using (\ref{hsxb}) and  (\ref{hsxbbis}),
It${\hat {\rm o}}$'s calculus and the commutation and anti-commutation properties of Pauli matrices.
Assuming unitary normalization, one can set $ |\boldsymbol \psi(t)|^2 =1$.

\subsection{The process: cartesian coordinates}

Given that the wave function is a complex object, the equality in (\ref{hsx2})
is a system of four coupled linear equations for four real variables
(the real and imaginary parts of $\psi_l $ and $\psi_r$).

We have seen that $\boldsymbol \psi^+\!(t)  \boldsymbol \psi (t)
=|\boldsymbol \psi(t)|^2 = |\psi_l(t)|^2+|\psi_r(t)|^2=1$ 
at all times, therefore, the number of really independent variables is at most three,
accordingly, we would like to rewrite (\ref{hsx2}) as a system of three equations for
the same number of variables. To reach this goal we define
 $ \bf x =\boldsymbol \psi^+ \! \boldsymbol {\hat \sigma} \boldsymbol \psi $,
 where $\boldsymbol {\hat \sigma} = \RV{(\hat\sigma_x ,\hat\sigma_y , \hat\sigma_z )}$. By components:

\begin{equation}
\begin{aligned}
x&=\boldsymbol \psi^+ \!  \hat \sigma_x \boldsymbol \psi 
= 2 {\rm Re } (\psi_l^* \psi_r), \\    
y&=\boldsymbol \psi^+ \! \hat \sigma_y \boldsymbol \psi 
= 2 {\rm Im } (\psi_l^* \psi_r) , \\
z&=\boldsymbol \psi^+ \!  \hat \sigma_z \boldsymbol \psi 
 = |\psi_l|^2 -|\psi_r|^2,
\end{aligned} 
\label{xyz}
\end{equation}
where the variable $z$ encodes the information abut localization.
In fact, given that $|\psi_l|^2 +|\psi_r|^2 =1$  one has $z \in [-1,1]$, therefore,
when $z=1$ the system is localized in the left well ($|\psi_l|^2= 1, |\psi_r|^2 =0$)
and  when  $z=-1$ it is localized in the right well ($|\psi_l|^2= 0, |\psi_r|^2 =1$).

By means of equations (\ref{hsxb}) and  (\ref{hsxbbis}), 
by the commutation and anti-commutation properties of the Pauli matrices 
and by It${\hat {\rm o}}$'s calculus one obtains the following equations:
\begin{equation}
\begin{aligned}
dx&=-2\beta^2 x\,dt-2\beta y \, dW_t , \\    
dy&=-2\beta^2 y\,dt-2\alpha z \,dt+2\beta x\, dW_t, \\
dz&=2\alpha y\, dt,
\end{aligned} 
\label{meq}
\end{equation}
which represent the superposition of two rotations, one (stochastic) around
the $z$ axis and another (deterministic) around the the $x$ axis.
As a consequence, the motion is confined on the surface of a sphere of 
constant radius. In fact, by It${\hat {\rm o}}$'s calculus one gets
\begin{equation}
\begin{aligned}
dx^2&= 2x dx + (dx)^2= -4\beta^2 x^2 \,dt-4\beta xy \, dW_t +  4 \beta^2 y^2 dt, \\    
dy^2&=2y dy+ (dy)^2=-4\beta^2 y^2 \,dt-4\alpha zy \,dt+4\beta xy \, dW_t +  4 \beta^2 x^2 dt, \\
dz^2&= 2zdz = 4\alpha zy \, dt.
\end{aligned} 
\label{eq: SDE2}
\end{equation}
It is immediate to verify that
\begin{equation}
d \rho^2 = d(x^2\!+\!y^2\!+\!z^2)=0,
\label{eq: SDE3}
\end{equation}
which means that the radius is constant. For later use it is also useful to write down the equation for the variable $zy$:
\begin{equation}
d(zy)= ydz+zdy=  2\alpha y^2 \, dt -2\beta^2 z y\,dt-2\alpha z^2 \,dt+2\beta xz\, dW_t, 
\label{eq: SDE4}
\end{equation}
in fact, if one takes the averages of (\ref{eq: SDE2}) and (\ref{eq: SDE4}) one ends-up with 
a closed system of differential deterministic equations.

In Appendix \ref{app:appA} we also provide an alternative two-variable formulation
of \eqref{meq} either adopting spherical coordinates or using
the spin conservation.
\subsection{Averages and interpretation}

We compute some averages which may highlight the behavior of the system.
From the stochastic equations (\ref{meq}) we obtain the following system 
of differential equations for the averaged variables:
\begin{equation}
\frac {d }{dt} \E[ x] =-2\beta^2  \E[  x]  , \,\,\,\,\,\,\,\, \,\,\,
\frac {d}{dt} \E[  y]  =-2\beta^2  \E[  y]   -2\alpha  \E[\RV{z}]   , \,\,\,\,\,\,\,\,\, \,\,\,
\frac {d }{dt} \E[  z]  =2\alpha  \E[  y],
\label{aver}
\end{equation}
where $ \E[  \cdot]  $ denotes the average with respect to the noise.
The characteristics of the solution of this system changes depending on 
the values of the parameters. When $\beta^2 < 2 | \alpha| $
(the small noise region) the solution, given a determined initial conditions
 $ \E[   {\bf x} (0) ]= {\bf x} (0) $, is
\begin{equation}
\begin{aligned}
 \E[  x(t)]  & = e^{-2 \beta^2 t} x(0),\\    
 \E[  y(t)]  & = e^{- \beta^2 t} [ y(0) \cos(\omega t)+ c_1  \sin (\omega t) ],\\
 \E[  z(t)]  & = e^{- \beta^2 t} [z(0) \cos(\omega t)+ c_2  \sin (\omega t) ],
\end{aligned} 
\label{solaver}
\end{equation}
where 
\begin{equation}
\omega  = |4 \alpha^2 - \beta^4|^{\frac 1 2}, 
\,\,\,\,\,\,\,\, \,\,\,
c_1= -\frac{1}{\omega} \, [ \beta^2 y (0) +2 \alpha  z (0)]
\,\,\,\,\,\,\,\,\, { \rm and } \,\,\,\,\,\,\,\,
c_2= \frac{1}{\omega} \, [ \beta^2 z (0) +2 \alpha y (0)]. 
\label{param}
\end{equation}
Looking at the third component z(t), one realizes that in this \red{small noise} region a
quantum coherent behavior survives the noise. The localization probability
is periodic with exponential damping: the system jumps from one state to
the other almost periodically. The damping factor, in fact, is a consequence
of the accumulation of errors due to the deviations from the purely periodic
behavior. The damping rate, that is $\beta^2$ for $y$ and $z$
and it is $2\beta^2$ for $x$,  increases with noise.

When $\beta^2 > 2 | \alpha| $ the solution is purely exponential and it may be obtained
from (\ref{solaver}) with the substitutions 
$\cos(\omega t) \to \cosh(\omega t)$, $\sin(\omega t) \to \sinh(\omega t)$.
In this large noise region the coherent behavior is completely shattered
since the localization probability relaxes exponentially. Therefore, the
system jumps randomly from one state to the other and it loses memory 
of its initial position with (large) rate $\beta^2$ for $x$ and with 
(small) rate $ \beta^2-|4 \alpha^2 - \beta^4|^{\frac 1 2}$  for $y$ and $z$.
One has
\begin{subequations}
\label{solaver2}
\begin{align}
 \E[  x(t)]  & = e^{-2 \beta^2 t} x(0) ,\label{eqxav}\\    
 \E[  y(t)]  & = e^{- \beta^2 t} [ y(0) \cosh(\omega t)+ c_1  \sinh (\omega t) ]
\simeq  -\frac{\alpha^2}{\beta^4} y(0) e^{-\frac{2 \alpha^2}{\beta^2}t},\label{eqyav}\\
 \E[  z(t)]  & = e^{- \beta^2 t} [z(0) \cosh(\omega t)+ c_2  \sinh (\omega t) ]
\simeq  z(0) e^{-\frac{2 \alpha^2}{\beta^2}t},\label{eqzav}
\end{align} 
\end{subequations}
where \RV{$\simeq$ denotes an approximation up to the leading order term in the extremely large noise regime $\beta^2 \gg 2|\alpha|$}.

If one compares the third equation in (\ref{solaver}) with the third in (\ref{solaver2}) 
one can appreciate the fact that in the small noise regime the localization is lost at a rate which increases with $\beta$,
but after transition to the large noise regime it decreases with $\beta$. 
Moreover, in the large noise regime, the memory of the initial values $x(0)$ and $y(0)$ is rapidly lost, 
while the memory of the initial value $z(0)$ persists longer.
This means that for large but finite $\beta$, 
the system experiences a fast distribution over the parallel fixed 
by the initial condition $z(0)$ followed by a slow distribution
over the other latitudes.
In the extremely large noise limit ($\beta \to \infty$) the system  
distributes instantaneously and uniformly over the parallel determined by the initial condition $z(0)$
(in fact, $\bar x (t) =\bar y(t) =0$ at any $t>0$),
but it does not change its latitude ($\E[z(t)] =z(0)$ at any $t \ge 0$).
Hence, the system remains localized if it is initially localized (i.e., \( z(0) = \pm 1 \)). This is precisely the phenomenology the model aims to illustrate.

\RV{Henceforth, it is convenient to adopt the following notation: given a random variable $X(t)\in \mathbb{R}^3$, we define the vector of mean values
\begin{equation}
m(t)=\E[X(t)]\in \R^{3},
\label{mean0}
\end{equation}
and the second-moment matrix
\begin{equation}
P(t)=\E[X(t)X(t)^T]\in \R^{3\times 3}.
\label{covar0}
\end{equation}
}
\subsection{\RV{The second-moment} matrix}
\label{sec:secmom}

\RV{A direct application of \^{I}to calculus (see Eq. \eqref{eq: covariance matrix} in Appendix~\ref{sec: GBMs} for details) shows that the evolution of the second-moment matrix $P$ is governed by the following system of ODEs:}

\begin{equation}
\begin{aligned}
 \dot{p}_{xx}&=-4\beta^2 (p_{xx}-p_{yy}),\\
 \dot{p}_{xy}&=-8 \beta^2 p_{xy}-2\alpha p_{xz},\\
 \dot{p}_{xz}&=-2\beta^2 p_{xz}+2\alpha p_{xy},\\
 \dot{p}_{yy}&=-4\beta^2 (p_{yy}-p_{xx})-4\alpha p_{yz},\\
 \dot{p}_{yz}&=-2\beta^2 p_{yz}-2\alpha p_{zz}+2\alpha p_{yy},\\
 \dot{p}_{zz}&=4\alpha p_{yz}.
\end{aligned}
\label{coveq}
\end{equation}
The above system decouples into two smaller systems
\begin{equation}
\label{eq: 2ODEs}
 \frac{d}{dt}\begin{pmatrix}
     p_{xy}\\
     p_{xz}
 \end{pmatrix}=\begin{pmatrix}
     -8\beta^2&-2\alpha\\
     2\alpha&-2\beta^2
 \end{pmatrix}   \begin{pmatrix}
     p_{xy}\\
     p_{xz}
 \end{pmatrix},
\end{equation}
and
\begin{equation}
\label{eq: 4ODEs}
 \frac{d}{dt}\begin{pmatrix}
     p_{xx}\\
     p_{yy}\\
     p_{yz}\\
     p_{zz}
 \end{pmatrix}=M\begin{pmatrix}
     p_{xx}\\
     p_{yy}\\
     p_{yz}\\
     p_{zz}
 \end{pmatrix},\quad M=\begin{pmatrix}
  -4\beta^2&4\beta^2&0&0\\
  4\beta^2&-4\beta^2&-4\alpha &0\\
  0&2\alpha&-2\beta^2&-2\alpha\\
  0&0&4\alpha&0
 \end{pmatrix},
 \end{equation}
where this second system is the one of physical interest since it contains the 
information concerning $p_{zz}(t) =\E[z^2(t)]$ which together with $\RV{m_z(t)=}\E[z(t)]$ (see \eqref{solaver2}) contains the information concerning
the localization of the system. The characteristic polynomial and conditions for the eigenvalues of $M$ to be all negative are
detailed in Appendix \ref{app:app2}.



\subsection{Stationary \RV{second-moment} matrix}
A stationary solution $\mathbf{p}_
\infty=(p_{xx}^\infty,p_{xy}^\infty,p_{xz}^\infty,p_{yy}^\infty,p_{yz}^\infty,p_{zz}^\infty)$ solves $M \mathbf{p}_\infty=0$, that is the following linear system
\begin{equation*}
\begin{cases}
 -4\beta^2 (p^\infty_{xx}-p^\infty_{yy})=0,\\
 -8 \beta^2 p^\infty_{xy}-2\alpha p^\infty_{xz}=0,\\
 -2\beta^2 p^\infty_{xz}+2\alpha p^\infty_{xy}=0,\\
 -4\beta^2 (p^\infty_{yy}-p^\infty_{xx})-4\alpha p^\infty_{yz}=0,\\
 -2\beta^2 p^\infty_{yz}-2\alpha p^\infty_{zz}+2\alpha p^\infty_{yy}=0,\\
 4\alpha p^\infty_{yz}=0,
\end{cases}
\end{equation*}
to be considered together with the condition, which follows from \eqref{eq: SDE3} and the assumption of unitary normalization (see the last paragraph in Section \ref{sec: two state system}):
\[
p^\infty_{xx}+p^\infty_{yy}+p^\infty_{zz}=1.
\]
It is a trivial task to verify that the solutions are 
\[
p^\infty_{xx}=p^\infty_{yy}=p^\infty_{zz}=\frac{1}{3},\quad p^\infty_{xy}=p^\infty_{xz}=p^\infty_{yz}=0.
\]
Thus, \RV{when $M$ has no eigenvalues with positive real part (in particular, see Appendix \ref{app:app2} for a condition under which all non-zero eigenvalues of $M$ are real and negative)} as $t\rightarrow +\infty$:
\[
\lim_{t\rightarrow+\infty}p_{xx}=\lim_{t\rightarrow+\infty}p_{yy}=\lim_{t\rightarrow+\infty}p_{zz}=1/3, \quad \lim_{t\rightarrow+\infty} p_{xy}=\lim_{t\rightarrow+\infty} p_{xz}=\lim_{t\rightarrow+\infty} p_{yz}=0.
\]

In the following sections, we analyze model reduction techniques for the geometric Brownian process described by Eq.~\eqref{meq}. Motivated by the physical significance of the variable \( z \), which captures localization information of the process, we develop a reduced description by preserving \( z \) as a resolved variable in the large noise regime. This reduction \RV{scheme can be implemented either through the adiabatic elimination or via the IM approach, leveraging the general framework outlined in Appendix~\ref{sec: GBMs}.
We first present the method of adiabatic elimination, which is rooted in the analysis of systems of higher-order ODEs.}

\section{Model reduction via higher-order ODEs}
\label{sec: higher order ODEs}

In this section, we first derive exact, closed-form equations for both the deterministic \red{part of} the dynamics and the \RV{evolution of the second moments}, which take the form of higher-order linear ODEs. Starting from these exact equations, we then construct corresponding reduced-order models by systematically neglecting higher-order terms. Furthermore, we demonstrate that these simplified models provide accurate approximations of the full dynamics within appropriate parameter regimes.
\RV{We begin by stating the following result, whose proof is deferred to Appendix~\ref{sec:HamCal}}.

\begin{restatable}{proposition}{HamCal}
\label{eq: general derivation}
Suppose that $F$ is an $n\times n$ matrix and $u=u(t)\in\R^n$ solves the following system of linear differential equations
\begin{equation}
\label{eq: eqnu}
\dot{u}=F u,\quad u(0)=u_0.
\end{equation}
Let $p_n(\lambda):=\sum_{i=0}^n a_i \lambda ^i$ be the characteristic polynomial of $F$. Then $u(t)$ satisfies the $n$-th order differential equations: 
 \begin{equation}
 \label{eq: high derivative eqn}
  \sum_{i=0}^n a_i u^{(i)}(t)=0,
 \end{equation} 
\RV{where $u^{(i)}(t)$, for $i=1,\ldots, n$, denotes the $i$-th derivative of $u(t)$.}
\end{restatable}

Suppose we now wish to derive a set of first-order ODEs for a subset of variables $\bar{u} \in \mathbb{R}^k$ that approximate the first $k$ components ($1 \le k \le n$) of $u$. \RV{The method of \textit{adiabatic elimination}~\cite{Risken} consists in neglecting the higher-order derivatives in Eq.~\eqref{eq: high derivative eqn}}, which leads to the reduced dynamics:
\begin{equation}
    a_1 \bar{u}'(t) + a_0 = 0.
    \label{eq:reduced_dynamics}
\end{equation}
\RV{This approach is employed in the subsequent sections to reduce the systems of ODEs governing the mean value and second-moment dynamics for the two-state quantum system introduced in Section~\ref{sec: two-state system}. Notably, such procedure yields a faithful approximation of the original evolution of the second moments in the large $\beta$ regime (as shown in Fig.~\ref{fig:fig1} below)}.
In what follows, to keep the notation light, and since each component $u_i$ satisfies the same differential equation sharing the structure of  Eq.~\eqref{eq:reduced_dynamics}, we will omit the index and refer to the generic component simply as $u$.


\subsection{\RV{Exact closed equations for the averages}}

The deterministic \red{component} of \RV{\eqref{meq}} obeys the set of ODEs given in \eqref{aver}, \RV{which we conveniently rewrite here in the form:}
\begin{subequations}
\begin{align}
\dot m_x &=-2\beta^2 m_x\;, \;\label{xdet} \\    
\dot m_y &=-2\beta^2 m_y-2\alpha m_z\;, \label{ydet}\\
\dot m_z &=2\alpha m_y\;. \label{zdet}
\end{align} 
\end{subequations}
The dynamics 
of $m_x$ is already governed by a closed equation. We will derive exact closed equations for $m_y$ and $m_z$ using Proposition \ref{eq: general derivation}. We have
\begin{equation}
\frac{d}{dt}\begin{pmatrix}
    m_y\\m_z
\end{pmatrix}=\begin{pmatrix}
    -2\beta^2 &-2\alpha\\
    2\alpha &0
\end{pmatrix}\begin{pmatrix}
    m_y\\m_z
\end{pmatrix}=:Q\begin{pmatrix}
    m_y\\m_z
\end{pmatrix}.
\label{subsysdet}
\end{equation}
The characteristic polynomial of $Q$ is given by
\[
|Q-\lambda I|=\det\begin{pmatrix}
    -2\beta^2 -\lambda& -2\alpha\\
    2\alpha&-\lambda
\end{pmatrix}=\lambda^2+2\beta^2\lambda +4\alpha^2.
\]
According to Proposition \ref{eq: general derivation}, both $m_y$ and $m_z$ satisfy the same higher-order differential equations
\begin{equation}
\label{eq: exact deterministic y and z}
  \ddot{u}+2\beta^2 \dot{u}+4\alpha^2 u=0,
\end{equation}
which can also be straightforwardly derived directly from (\ref{subsysdet}).

\subsubsection{\RV{Reduced equation for the average: large $\beta$ regime}}

In the regime $\beta^2 \gg 2|\alpha|$, \RV{we use the method of adiabatic elimination to omit} the second-order terms in the exact second-order equation \eqref{eq: exact deterministic y and z} for $y$ and $z$ and obtain the following reduced dynamics for $x,y,z$
\begin{subequations}
\begin{align}
  \dot{\bar x}&=-2\beta^2 \bar x,\label{rxdet}\\
  \dot{\bar y}&=-\frac{2\alpha^2}{\beta^2}\bar y,\label{rydet}\\
  \dot{\bar z}&=-\frac{2\alpha^2}{\beta^2}\bar z.\label{rzdet}
\end{align}
\end{subequations}
The corresponding solution for $\bar{z}$ is thus given by
\[
\bar{z}(t)=\bar{z}(0)e^{-\frac{2\alpha^2}{\beta^2}t},
\]
\RV{which recovers Eq. \eqref{eqzav}.}
\subsection{Exact closed equations for the \RV{second moments}}
We recall from Sec.~\ref{sec:secmom} that the \RV{second moments} $p_{xx}, p_{yy}, p_{yz}$ and $p_{zz}$ satisfy the following system of linear ODEs
\[
\frac{d}{dt}\begin{pmatrix}
    p_{xx}\\
    p_{yy}\\
    p_{yz}\\
    p_{zz}
\end{pmatrix}=M\begin{pmatrix}
    p_{xx}\\
    p_{yy}\\
    p_{yz}\\
    p_{zz}
\end{pmatrix},
\]
where $M$ is given in \eqref{eq: 4ODEs}. The characteristic polynomial of $M$ reads:
\[
|M-\lambda I|=\lambda^4 + 10\beta^2\lambda^3 + (16\beta^4 + 16\alpha^2)\lambda^2 + 96\alpha^2\beta^2\lambda.
\]
see Appendix \ref{app:app2}.
Applying Proposition \ref{eq: general derivation}, we deduce that all the \RV{second moments} satisfy the same $4$-th order differential equation :
\begin{equation}\label{eq:pzz-fourth}
    \ddddot{u}+10\beta^2 \dddot{u}+16(\alpha^2+\beta^4)\ddot{u}+96\alpha^2\beta^2 \dot{u}=0,
\end{equation}
or equivalently
\begin{equation}
\label{eq: pzz C}    
\dddot{u}+10\beta^2 \ddot{u}+16(\alpha^2+\beta^4)\dot{u}+96\alpha^2\beta^2 u=C,
\end{equation}
where $C$ is a constant depending on $u$. It can be obtained by letting $t$ in the \RV{left-hand side} to infinity, yielding $C=96\alpha^2\beta^2 u_\infty$. Since $p_{xx}^\infty=p_{yy}^\infty=p_{zz}^\infty=1/3$, we obtain that $p_{xx}, p_{yy}$ and $p_{zz}$ all satisfy the third-order differential equation
\begin{equation}
\label{eq: exact equation for variances}   
\dddot{u}+10\beta^2 \ddot{u}+16(\alpha^2+\beta^4)\dot{u}+96\alpha^2\beta^2 (u-1/3)=0. 
\end{equation}
In Appendix~\ref{sec: exact closed explicit} we provide a direct approach to derive the higher-order differential equation for $p_{zz}$ which is the quantity of interest in the model reduction scheme pursued in this paper.

\subsubsection{Reduced equations for the \RV{second moment}: large $\beta$ regime}

We now turn to the large $\beta$ regime  where $\alpha$ is fixed and $\beta^2 \gg 2 |\alpha|$. In this regime, the term involving the third-order derivative in the exact equation \eqref{eq: exact equation for variances} becomes negligible. \RV{Applying the adiabatic elimination, we neglect} this term and get the following reduced, second-order linear ODEs governing the evolution of all the \RV{second moments}:
\begin{equation}
\label{eq: high temperature regime second order}
5\beta^2 \ddot{u}+8(\alpha^2+\beta^4)\dot{u}+48\alpha^2\beta^2 (u-1/3)=0.
\end{equation}
We can simplify further: first by ignoring the term $8\alpha^2 \dot{u}$ (since $|\alpha| \ll \beta^2$) and dividing by $\beta^2$, we get
\[
5\ddot{u}+8\beta^2\dot{u}+48 \alpha^2(u-1/3)=0.
\]
Then we remove the second-order term, to obtain a first order differential equation:
\begin{equation}
\label{eq: high temperature regime first order1}    
\beta^2\dot{u}+6 \alpha^2(u-1/3)=0.
\end{equation}
Note that we have kept the last term to preserve the equilibrium value. We take \eqref{eq: high temperature regime first order1} as the reduced equation for $p_{zz}$. Thus, denoting by $\bar{p}_{zz}$ the adiabatic approximation of the second moment, we get
\begin{equation}
\label{eq: high temperature regime first order}    
\dot{\bar{p}}_{zz}(t)=-\frac{6\alpha^2}{\beta^2}\Big(\bar{p}_{zz}(t)-1/3\Big).
\end{equation}
The explicit solution to \eqref{eq: high temperature regime first order} is given by
\begin{equation}
\bar{p}_{zz}(t)=\frac{1}{3}+\left(u(0)-\frac{1}{3}\right)e^{-\frac{6\alpha^2}{\beta^2} t}.  
\label{eq:red-sol}
\end{equation}


The graphs of the solution to Eq. \eqref{eq: exact equation for variances}, supplied with the initial data \( p_{zz}(0) =  p_{zz}'(0) = p_{zz}''(0) = 0  \), and of Eq. \eqref{eq:red-sol} are shown in Fig. \ref{fig:fig1}, for different values of $\beta$, in the regime where $\beta^2 > 2 |\alpha|$. In particular, the graph of the solution to Eq. \eqref{eq: exact equation for variances} coincides with the corresponding solution to the original set of ODEs \eqref{eq: 4ODEs}. 
The left panel also evidences that, if $\beta$ is not large enough, the adiabatic elimination result can deviate significantly from the solution of the original ODE system. 
The two panels shown in Fig. \ref{fig:fig1} should be compared to the corresponding panels of Fig. \ref{fig:red3b} in Sec. \ref{sec:redvar}.

\begin{figure}[H]
    \centering     
    \includegraphics[width=0.45\linewidth]{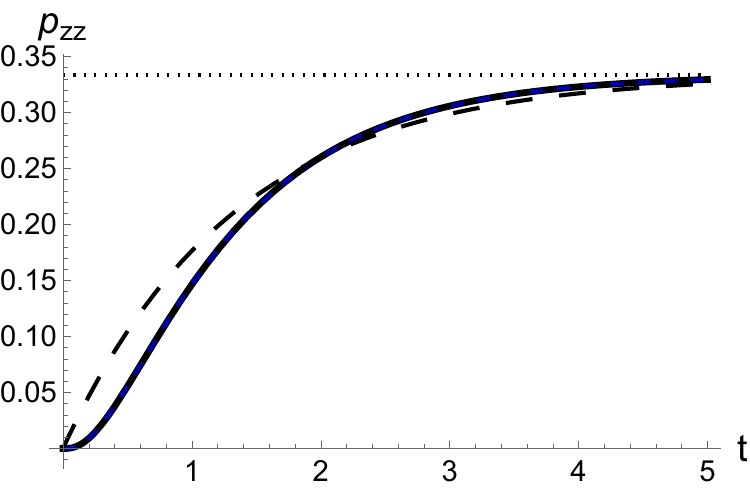}
    \includegraphics[width=0.45\linewidth]{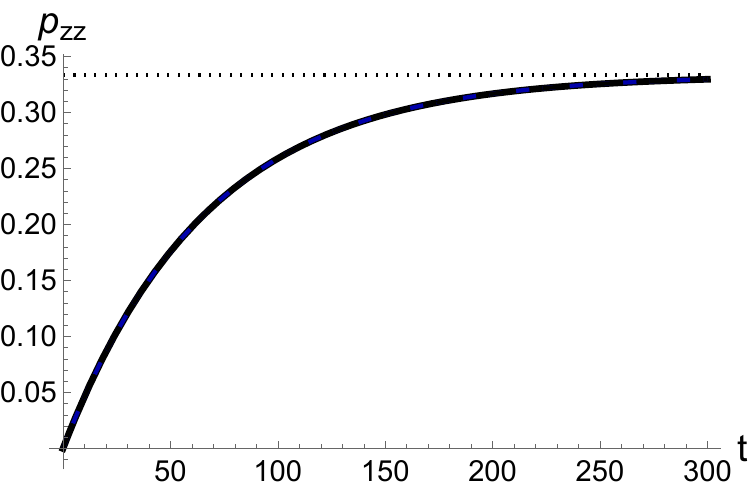}
    \caption{Time evolution of \( p_{zz}(t) \), computed from the the third-order ODE \eqref{eq: exact equation for variances} (black solid line) and from the adiabatic elimination \eqref{eq: high temperature regime first order} (black dashed line), with initial conditions \( p_{zz}(0) = 0, p_{zz}'(0) = 0, p_{zz}''(0) = 0  \). The dotted horizontal line marks the asymptotic value $1/3$. The parameters are set to the values \( \alpha = 0.5 \) and
 \( \beta^2 = 2 \) (left panel) and \( \beta^2 = 100 \) (right panel).}
    \label{fig:fig1}
\end{figure}

In Section~\ref{sec: IM method}, we contrast this strategy with an alternative reduction technique based on the IM method, which is largely employed in the analysis of dynamical systems characterized by multiple time scales.

\section{Model reduction via the IM method}
\label{sec: IM method}
The basic idea underlying the IM method is that, after a brief transient, the phase space dynamics of the system becomes confined to a lower-dimensional surface embedded within the full phase space~\red{\cite{fenichel1979geometric, GorKar05}}. This surface is invariant under the flow, i.e. any trajectory that starts on it remains confined there for all future times. In many applications, this invariant surface corresponds to the so-called slow manifold, which governs the long-term evolution of the system once the fast modes have relaxed.
In our setup, the implementation of the method begins with a rescaling of the time variable \( t \) and the parameter \( \beta \), as discussed in the following section.

\subsection{Rescaled systems}

Following the approach developed in \cite{colangeli2025hybrid}, we rescale $t\mapsto \epsilon t$ and $\beta\mapsto \beta/\sqrt{\epsilon}$. Under this time rescaling, the Wiener process transforms as $dW_{\epsilon t} = \sqrt{\epsilon}dW_t$, leading to the system:
\begin{subequations}
\label{eq: mainSDE}
\begin{align}
dx&=-2\beta^2 x\,dt-2\beta y \, dW_t\, ,\label{x}\\ 
dy&=-2\beta^2 y\,dt-2 \epsilon\alpha z \,dt+2\beta x\, dW_t\, ,\label{y}\\
dz&=2 \epsilon\alpha y\, dt \, ,\label{z}
\end{align} 
\end{subequations}
which describes the original dynamics observed on the \textit{fast} time scale.

    \subsection{The deterministic \red{part} of the rescaled system}


The deterministic \red{component} of the original dynamics for the variables $y$ and $z$ takes then the form

\begin{equation}
\frac{d}{dt}\begin{pmatrix}
    m_y\\m_z
\end{pmatrix}=\begin{pmatrix}
    -2\beta^2 &-2\epsilon\alpha\\
    2\epsilon\alpha &0
\end{pmatrix}\begin{pmatrix}
    m_y\\m_z
\end{pmatrix}=:Q_{\epsilon}\begin{pmatrix}
    m_y\\m_z
\end{pmatrix}.
\label{eq: deterministicODE}
\end{equation}

\subsection{The \RV{second-moment} equations of the rescaled system}
Exploiting the same scaling $\beta \mapsto\beta/\sqrt{\epsilon}$
and $t\mapsto \epsilon t$ with the system 
\eqref{eq: 4ODEs}, we obtain:


\begin{equation}
\frac{d}{dt}\begin{pmatrix}
    p_{xx}\\
    p_{yy}\\
    p_{yz}\\
    p_{zz}
\end{pmatrix}=M_{\epsilon}\begin{pmatrix}
    p_{xx}\\
    p_{yy}\\
    p_{yz}\\
    p_{zz}
\end{pmatrix} \; ,
\label{eqcov}
\end{equation}
with
\begin{equation}
 M_{\epsilon}=\begin{pmatrix}
  -4\beta^2&4\beta^2&0&0\\
  4\beta^2&-4\beta^2&-4\epsilon\alpha &0\\
  0&2\epsilon\alpha&-2\beta^2&-2\epsilon\alpha\\
  0&0&4\epsilon\alpha&0
 \end{pmatrix}  \; . 
 \label{Meps}
\end{equation}

\subsection{Reduction of the deterministic \red{component}}
\label{sec:IMdet}


We introduce a closure relation of the form $m_y(t) = a(\alpha,\beta,\epsilon)m_z(t)$. \red{Substituting this relation into \eqref{eq: deterministicODE}, we can compute $\dot{m}_y$ in two ways
\begin{align*}
    \dot{m}_y= a\dot{m}_z= 2\epsilon \alpha a m_y\quad\text{and}\quad \dot{m}_y=-2\beta^2 m_y-2\epsilon\alpha m_z=\Big(-2\beta^2 -2\frac{\epsilon\alpha}{a}\Big)m_y,
\end{align*}
from which we deduce the following invariance equation for the function $a(\alpha,\beta,\epsilon)$:}
\begin{equation}
    \epsilon\alpha a^2 + \beta^2 a + \epsilon\alpha = 0 \, .
    \label{IE1}
\end{equation}

The solutions to the quadratic equation \eqref{IE1} are given by:
\begin{equation}
    a_{\pm}(\varepsilon) = \frac{-\beta^2 \pm \sqrt{\beta^4 - 4\epsilon^2\alpha^2}}{2\epsilon\alpha} \, ,
    \label{IE1sol}
\end{equation}
\RV{where the dependence of $a_{\pm}$ on $\alpha$ and $\beta$ has been suppressed for brevity}. These roots remain real-valued provided that $\epsilon \leq \epsilon_c' = \beta^2 / (2|\alpha|)$. Notably, the roots $a_{\pm}$ reconstruct the two eigenvalues of the matrix $Q_{\epsilon}$ in \eqref{eq: deterministicODE}, as discussed in Appendix \ref{app:app4}. Specifically, the root $a_{+}$ corresponds to the eigenvalue of $Q_{\epsilon}$ with the smaller magnitude, which vanishes in the limit $\epsilon \to 0$. Consequently, the reduced description for $m_z(t)$ takes the form:
\begin{equation}
    \dot{m}_z(t)=\RV{2\varepsilon \alpha a_{+}(\varepsilon)m_z(t) }=-(\beta^2-\sqrt{\beta^4-4 \varepsilon^2 \alpha^2})m_z(t) 
    \label{redex2}     \; .
\end{equation}
Enforcing $\epsilon \ll \epsilon_c'$ imposes a large time-scale separation between fast and slow variables, which corresponds to the regime where adiabatic elimination is applicable. Indeed, expanding \RV{$2\varepsilon \alpha a_{+}(\varepsilon)$} in powers of $\epsilon$ yields
\begin{equation}
    \RV{2\varepsilon \alpha a_{+}(\varepsilon)} = -\frac{2\alpha^2 \epsilon^2}{\beta^2} + o(\epsilon^2) \, ,
    \label{match}
\end{equation}
\RV{which, upon further rescaling $t\mapsto t/\epsilon$ and $\beta\mapsto \sqrt{\epsilon}\beta$, allows \eqref{redex2} to recover \eqref{rzdet}}. 
This result underscores a key feature of the IM method: it not only subsumes adiabatic elimination as a special case, but also extends the validity of the reduced description across the entire interval $\epsilon \in (0, \epsilon_c']$. Notably, the reduced description breaks down at $\epsilon=\epsilon_c'$, where the two eigenvalues of $Q_{\epsilon}$ merge and, for $\epsilon > \epsilon_c'$, become complex-valued. This transition corresponds to the loss of the invariant manifold structure and marks the limit of applicability for the reduction procedure.

\subsection{Reduction of the variance $p_{zz}(t)$}
\label{sec:redvar}

To comply with the structure of the ODE system \eqref{eqcov}, we introduce the following (affine) closures:
\begin{equation*}
    p_{xx}(t)=a_1 \left(p_{zz}(t)-\frac{1}{3}\right)+\frac{1}{3}\; , \; p_{yy}(t)=a_2 \left(p_{zz}(t)-\frac{1}{3}\right)+\frac{1}{3}\; , \; p_{yz}(t)=a_3 \left(p_{zz}(t)-\frac{1}{3}\right)\; ,
\end{equation*}
with yet unknown coefficients $a_1,a_2,a_3\in \mathbb{R}$. This closure ansatz ensures the correct equilibrium values $p_{xx}^\infty=p_{yy}^\infty=p_{zz}^\infty=1/3$.

The invariance equations for the coefficients \( a_1, a_2, a_3 \) amount to the following system of algebraic equations:
\begin{subequations}
   \label{IE2} 
   \begin{align}
-4\beta^2 a_1 + 4\beta^2 a_2 - 4\epsilon\alpha a_1 a_3 &= 0 \; , \\ 
4\beta^2 a_1 - 4\beta^2 a_2 - 4\epsilon\alpha a_3 - 4\epsilon\alpha a_2 a_3 &= 0 \; , \\  
2\epsilon\alpha a_2 - 2\beta^2 a_3 - 2\epsilon\alpha - 4\epsilon\alpha a_3^2 &= 0 \; .
  \end{align}
\end{subequations}

We examine the eigenvalues of the matrix \(M_{\epsilon} \), defined in Eq.~\eqref{Meps}, which are plotted in Fig.~\ref{fig:eigen} as functions of \( \epsilon \). We verify numerically that the spectrum of \(M_{\epsilon} \) undergoes a qualitative change at a critical value $\epsilon_c''$, which signals the breakdown of the reduced description: for \( \epsilon < \epsilon_c'' \), all eigenvalues remain real and distinct, whereas for \( \epsilon > \epsilon_c'' \) a pair of conjugate complex eigenvalues appears. 
Among the solution branches $\{a_1(\epsilon), a_2(\epsilon), a_3(\epsilon)\}$ satisfying the invariance equations \eqref{IE2}, we select those that are both continuous in $\epsilon$ and fulfill the  asymptotic condition:
\begin{equation}  
\lim_{\epsilon \to 0} \epsilon a_3(\epsilon) = 0 \; ,
\label{asympt}
\end{equation}
which guarantees that the coefficient $4\alpha\epsilon a_3^*(\epsilon)$ appearing in the reduced dynamics for $p_{zz}(t)$ maintains the proper asymptotic behavior in the small $\epsilon$ limit.
Numerically, we identify two such solutions. The first is the trivial solution \( (a_1, a_2, a_3) = (1,1,0) \), which is disregarded as it is independent of \( \epsilon \). The second, denoted \( (a_1^*(\epsilon), a_2^*(\epsilon), a_3^*(\epsilon)) \), is the physically relevant one.

\begin{figure}[H]
    \centering   
    \includegraphics[width=0.45\linewidth]{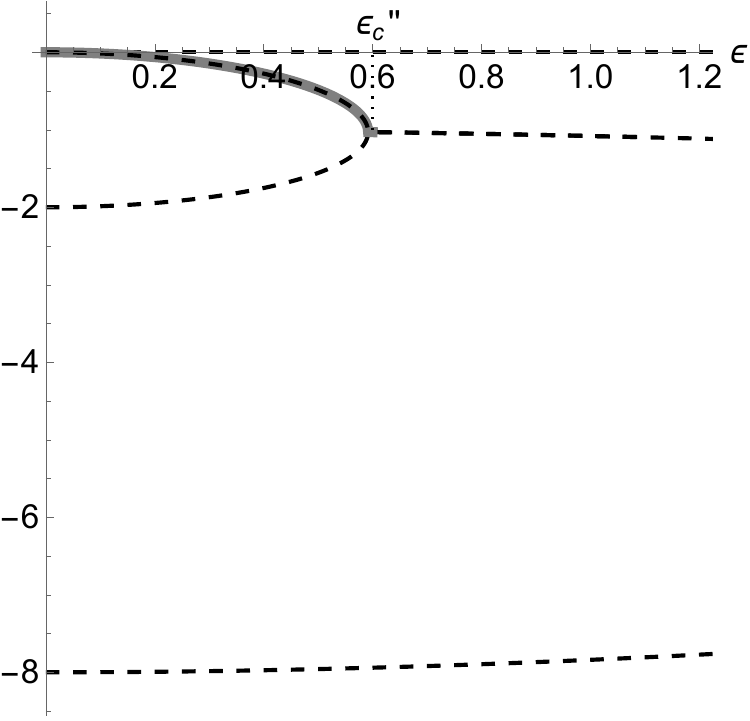}    
    \caption{
Behavior of the eigenvalues of the matrix \(M_{\epsilon} \) in Eq.~\eqref{Meps} (black dashed lines) and \( 4\alpha \epsilon a_3^*(\epsilon) \) (solid gray line) as functions of $\epsilon$, using $\alpha=0.5$ and $\beta=1$. The critical point \( \epsilon_c'' \simeq 0.59185 \) marks the onset of a pair of complex conjugate eigenvalues. 
}
    \label{fig:eigen}
\end{figure}

Restoring the original (slow) time scale through the transformation $t \mapsto t/\epsilon$, the reduced dynamics for \( p_{zz}(t) \) takes the form:
\begin{equation}
    \dot{\bar{p}}_{zz}(t) = 4\alpha a_3^*(\epsilon)\left(\bar{p}_{zz}(t) - \frac{1}{3}\right) \; ,
    \label{exred}
\end{equation}
which should be compared with the result obtained from the adiabatic elimination expressed by Eq.~\eqref{eq: high temperature regime first order}, rewritten here under the rescaling \( \beta \mapsto\beta/\sqrt{\epsilon} \):
\begin{equation}
    \dot{\bar{p}}_{zz}(t) = -\frac{6\alpha^2 \epsilon}{\beta^2}\left(\bar{p}_{zz}(t) - \frac{1}{3}\right) \; .
    \label{adiab}
\end{equation}
Remarkably, a Chapman--Enskog method of solution can be systematically applied to the invariance equations \eqref{IE2} through an expansion of the coefficients $a_1,a_2,a_3$ in powers of $\epsilon$. The leading-order analysis yields the relation:
\begin{equation}
   4\alpha a_3^*(\epsilon) = -\frac{6\alpha^2 \epsilon}{\beta^2} + o(\epsilon) \; ,
\end{equation}
which demonstrates that the IM approach correctly reproduces the adiabatic result in the small-$\epsilon$ regime while remaining valid for finite values of $\epsilon$ as well. The full behavior across the large-$\beta$ regime is illustrated in Fig.~\ref{fig:red3b}.

\begin{figure}[H]
    \centering     
    \includegraphics[width=0.45\linewidth]{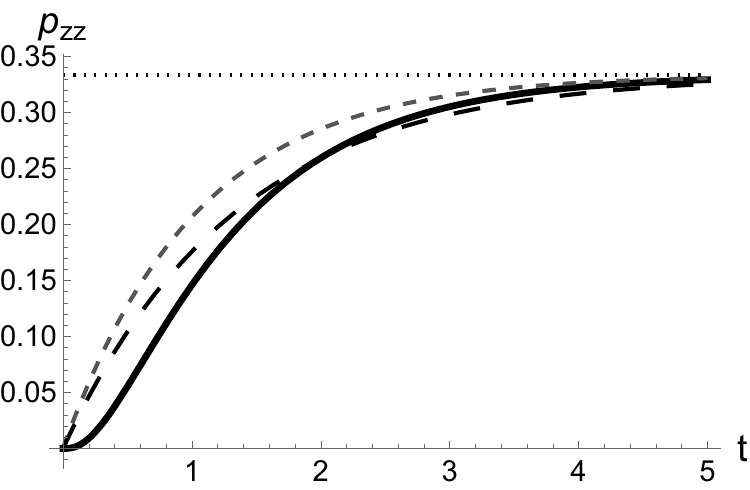}
    \includegraphics[width=0.45\linewidth]{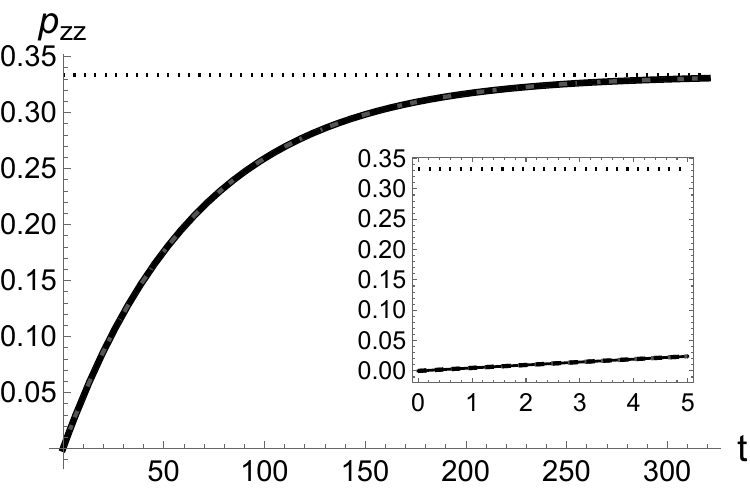}
    \caption{Time evolution of \( p_{zz}(t) \) computed from the original system \eqref{eqcov} (black solid line), and time evolution of the reduced dynamics $\tilde{p}_{zz}(t)$ computed from the invariant manifold reduction \eqref{exred} (gray short-dashed line) and from the adiabatic elimination \eqref{adiab} (black dashed line), with initial condition \( p_{zz}(0) = 0 \). Results are shown for \( \epsilon = 0.5 < \epsilon_c'' \) (left panel) and \( \epsilon = 0.01 \) (right panel). The inset in the right panel offers a magnified view of the early-time dynamics. The dotted horizontal line marks the asymptotic value \( 1/3 \). In both panels, parameters are set to \( \alpha = 0.5 \) and \( \beta = 1 \).}
    \label{fig:red3b}
\end{figure}

As a final remark, we observe that the onset of critical values $\epsilon_c'$ and $\epsilon_c''$ imply that the reduction scheme based on the IM method works for $\epsilon\le \min\{\epsilon_c',\epsilon_c''\}$.

\section{Full SDE reduced model of $z$}
\label{sec:reddyn}
In this section, we utilize the general reduction scheme described in \RV{Appendix}~\ref{sec: GBMs} to derive a reduced stochastic differential equation for the variable $z$ from the full system of SDEs in Eq. \eqref{meq}. 
\subsection{A general model reduction of GBMs for 1D reduced dynamics}
\label{sec: reduction scheme}
Let \( n \geq 1 \) be an integer, and let \( A, B, D \in \mathbb{R}^{n \times n} \) be given matrices. We consider \( n \)-dimensional GBMs\footnote{More precisely, Eq.~\eqref{generalSDE} combines features of both an Ornstein–Uhlenbeck process and a geometric Brownian motion.} of the form:
\begin{equation}
    dX\RV{(t)}= AX\RV{(t)}\,dt+BX\RV{(t)}\,dW\RV{(t)}+D\,dU\RV{(t)},
    \label{generalSDE}
\end{equation}
with $W{(t)}$ and $U({(t)}$ denoting independent standard Wiener processes of dimension $1$ and $n$, respectively. The above SDEs are equipped with the initial data $X(0)=X_0$. \RV{The dynamics of the mean and second moments of \eqref{generalSDE} are reviewed in Appendix \ref{sec: means-2ndmoments}.}

Our aim is to construct a reduced dynamics for a subset of resolved variables. \RV{In this section, we focus on the case of a one-dimensional (1D) resolved variable, consistent with the application to the two-state quantum system discussed in Section~\ref{sec: two-state system}. A potential strategy for extending this approach to the multidimensional case is discussed in Appendix~\ref{sec: GBMs}.}

\RV{We denote the resolved variable by $X^{\rm red}\in \R$, which is assumed to obey a one-dimensional} mixed Ornstein–Uhlenbeck (OU) and geometric Brownian process:
\begin{equation}
  \label{eq: general reduced SDE}
  d X^{\rm red}\RV{(t)}= \bar A X^{\rm red}\RV{(t)}\, dt+\bar B X^{\rm red}\RV{(t)} dw(t)+\bar D du\RV{(t)},
\end{equation}
where $\bar A, \bar B, \bar D\in \RV{\R}$ are the reduced drift and diffusion coefficients, $w\RV{(t)}$ and $u\RV{(t)}$ are independent \RV{one-dimensional} Wiener processes. The advantage of a mixed OU and GBM is that it is fully determined by the dynamics of the mean and the second moments; from these, one can identify the corresponding drift and diffusion matrices. Our task is to determine the reduced drift and diffusion coefficients $\bar A, \bar B, \bar D\in \RV{\R}$.  
\RV{The goal of our model reduction approach is to construct a reduced GBM that captures the exact mean (deterministic
part of the dynamics) and the exact stationary second moment, while providing the closest evolution of the second-moment matrix. Guided by these principles,} we propose the following procedure for model reduction:
\begin{itemize}
    \item derive the reduced model for the deterministic part of \eqref{generalSDE} using either the adiabatic or IM approach:
    \[
    \dot{\bar{X}}=\bar{A} \bar{X}+\bar{a}.
    \]
    where $\bar{A}\in \RV{\R, \bar a\in \R}$.
    \item derive the reduced equation for the \RV{second moment}, again using either the adiabatic elimination or the IM approach:
    \[
    \dot{\bar{p}}= \bar{H} \bar{p}+\bar{h}.
    \]
    where the drift \RV{coefficient $\bar{H}\in\R, \bar{h}\in \R$} 
\item The diffusion \RV{coefficients $\bar{B}, \bar{D} \in \R$} satisfy the 
\RV{one-dimensional} Lyapunov equation \RV{(see Eq. \eqref{eq: stationary covariance matrix} in Appendix \ref{sec: means-2ndmoments} for the general case)}
\begin{equation}
 (2\bar A +\bar B^2) \bar{p}_\infty+\bar D^2=0 \; ,
\label{Lyapgen}
\end{equation}  
 where \RV{$\bar p_\infty\in \R$} is the \RV{part} of the original invariant matrix $P_\infty$ corresponding to the resolved variables. Equation \eqref{Lyapgen} establishes a connection between the diffusion \RV{coefficients} $\bar B, \bar D$, the drift \RV{coefficient} $\bar A$ and the stationary \RV{second-moment} $\bar p_\infty$. We emphasize that this condition generalizes the classical \textit{fluctuation-dissipation} relation for OU processes \cite{Kubo66,Zwanzig,Pavl,CDM22} to the broader case of GBMs with \RV{multiplicative} noise, as described by Eq.~\eqref{generalSDE}. This relation ensures that the reduced dynamics has the same marginal invariant measure as the original one. Since $\bar A$ is known from the first step, it remains to find either $\bar{B}$ or $\bar D$. We find $\bar B$ and $\bar D$ such that the \RV{second-moment} dynamics of the reduced dynamics \RV{(see \eqref{eq: linear ODE} in Appendix \ref{sec: means-2ndmoments} for details)
\begin{equation*}
\label{covariance1D}
\dot{\tilde{p}}=\bar{D}^2+(2\bar{A}+\bar{B}^2)\tilde{p},  
\end{equation*}}is closest to the reduced second-moment dynamics obtained in the second step (via the adiabatic elimination or the IM method) in a suitable norm $\|\cdot\|$, for instance the $L^\infty$-norm $\|f\|_\infty=\sup_{t\geq 0}|f(t)|$ or the $L^2$-norm, $\|f\|_2=\Big(\int_0^\infty f(t)^2\,dt\Big)^{1/2}$. That is, $\bar B$ and $\bar D$ solve the following minimization problem
\begin{equation}
\label{eq: general minimization problem}
\min_{\bar B,\bar D}\|\bar p -\tilde{p}\|.
\end{equation}
\RV{This optimization problem can be solved explicitly for both $L^2-$ and $L^\infty-$ norms, see Lemmas \ref{lem: lem1}-\ref{lem: lem2} below and Appendix \ref{app: L2norm}. In Appendix \ref{sec: multipleD}, we discuss possible approaches to solve the corresponding optimization problem when the reduced dynamics is multi-dimensional.}
\end{itemize}
The distance between the \RV{second-moment} $\tilde p$ of the reduced dynamics and the original one $p\vert_{\bar X}$ \RV{(which is the subpart of the second-moment matrix of the original dynamics that corresponds to the resolved variables)} can be estimated via
\begin{equation}
\label{eq: error estimate}    
\|\tilde p-p\vert_{\bar X}\|\leq \|\tilde p-\bar{p}\|+\|\bar{p}-p\vert_{\bar X}\|.
\end{equation}
\red{Estimates of errors between the original and reduced dynamics using the $L^2$-norm has been investigated previously in the literature, see for instance \cite{redmann2021bilinear} and references therein.} The key message is that $p\vert_{\bar X}$ is expensive to compute (thus the left-hand side), but the two terms on the right-hand side are simpler \red{because $\bar{p}$ is essentially an explicit truncation of $p\vert_{\bar X}$ }. This also explains why we want to minimize the last term by solving the minimization problem \eqref{eq: general minimization problem}.

\subsection{Derivation of the reduced models}
In the previous sections, using either the adiabatic elimination method in Section \ref{sec: higher order ODEs} or the IM method in Section \ref{sec: IM method}, we have already derived the following reduced dynamics for the deterministic \red{component} and \RV{second moment} of the variable $
z$:
\begin{itemize}
    \item reduced deterministic \red{part}, 
    \begin{equation}
\label{eq: bar Az}
\dot{\bar z}=\bar{A}_z\bar z \quad\text{where}\quad 
\bar{A}_z=\begin{cases}
-2\alpha^2/\beta^2 \quad \text{in the adiabatic elimination method, cf. \eqref{rzdet}},\\
-(\beta^2-\sqrt{\beta^4-4\alpha^2})\quad \text{in the IM method, cf. \eqref{redex2}}.
\end{cases}
    \end{equation}

\item reduced \RV{second moment}, 
\[
\dot{\bar{p}}_{zz}=\tilde{A}\left( \bar{p}_{zz}-\frac{1}{3}\right)\quad\text{where}\quad \tilde{A}=\begin{cases}
   -6\alpha^2/\beta^2\quad \text{in the adiabatic elimination method, cf. \eqref{eq: high temperature regime first order}},\\
 4\alpha a_3^*\quad \text{in the IM method, cf. \eqref{exred}}.
\end{cases}    
\]
\end{itemize}
By exploiting the general reduction scheme described above, we now construct a full reduced SDE for $z$, which is of the form of a general 1-dimensional geometric Brownian motion:
\begin{equation}
\label{eq:bar z}
d z^{\rm red}\RV{(t)}= \bar A_z z^{\rm red}\RV{(t)}\,dt+\bar B_z z^{\rm red}\RV{(t)}\, dw_t+\bar D_z\, du_t,    
\end{equation}
where $\bar{A}_z$ is the drift coefficient obtained from the reduced deterministic \red{part}, either using the adiabatic elimination method or the Invariant Manifold method, \red{and $w_t$ and $u_t$ are independent one-dimensional Wiener processes.
} To find $\bar B_z$ and $\bar D_z$ we require the following two conditions
\begin{itemize}
    \item $\bar B_z$ and $\bar D_z$ satisfy the following relation:
    \begin{equation}
\label{eq: red Lyapunov cond}
(2\bar{A}_z+\bar{B}_z^2) p^\infty_{zz}+\bar{D}_z^2=0.
    \end{equation}
This relation ensures that the reduced dynamics \eqref{eq:bar z} has the correct equilibrium $\bar z_\infty=z_\infty$.
\item The \RV{second moment} computed from the reduced dynamics \eqref{eq:bar z} is closest to the reduced \RV{second moment $\tilde{p}_{zz}$} in the $L^\infty$-norm defined below. 
\end{itemize}
\RV{The second-moment equation for $\bar z$ is given by (see Eq. \eqref{eq: covariance matrix} in Appendix \ref{sec: means-2ndmoments} for the general derivation)}
\[
\dot{\tilde{p}}_{zz}=(2\bar{A}_z+\bar{B}^2_z)(\tilde{p}_{zz}-p^\infty_{zz}),\quad \tilde{p}(0)=\tilde{p}_0, \quad \text{where}\quad p^\infty_{zz}=1/3.
\]
This has an explicit solution
\[
\tilde{p}_{zz}(t)=\RV{p_{zz}^\infty} +(\tilde{p}_0-p_{zz}^\infty)e^{(2\bar{A}_z+\bar{B}^2_z)t}.
\]

We aim to find $\bar{B}_z$ such that the above dynamics is closest to the reduced dynamics obtain\red{ed} from the adiabatic limit. To this end, we compute their distance based on the supremum norm, \RV{assuming that $\bar{p}_{zz}(0)=\tilde{p}_{zz}(0)=\tilde{p}_0$,}
\[
\|\bar{p}_{zz}-\tilde{p}_{zz}\|_\infty=\sup_{t\in[0,\infty)}|\bar{p}_{zz}(t)-\tilde{p}_{zz}(t)|=|\RV{(\tilde{p}_0-p^\infty_{zz})}|\, \sup_{t\in[0,\infty)}|(e^{\tilde{A} t}-e^{(2\bar{A}_z+\bar{B}^2_z)t})|,
\]
and solve the following minimization problem for $\bar B_z$:
\begin{equation}
\label{eq: min bar B}
\min_{B_z}\|\bar{p}_{zz}-\tilde{p}_{zz}\|_\infty.
\end{equation}
We address this optimization problem in two steps via the following two auxiliary lemmas \RV{whose proofs are given in Appendix \ref{app:app6}.}
\begin{lemma}
\label{lem: lem1}
Let $a,b<0$. Then
\[
\sup_{t\geq 0}|e^{at}-e^{bt}|=\Big(\frac{a}{b}\Big)^\frac{b}{b-a}-\Big(\frac{a}{b}\Big)^\frac{a}{b-a}.
\]
\end{lemma}
Let $a<0$. We consider the following function of $b$: 
\[
F(b):=\Big(\frac{a}{b}\Big)^\frac{b}{b-a}-\Big(\frac{a}{b}\Big)^\frac{a}{b-a}.
\]
\begin{lemma}
\label{lem: lem2}
 We have
 \[
 \min_{a<b_0\leq b< 0} F(b)=F(b_0).
 \]
\end{lemma}
Returning back to the model reduction of the variable $z$. We apply Lemmas \ref{lem: lem1} and \ref{lem: lem2} for
\[
a=\tilde{A},\quad b=2\bar{A}_z+\bar{B}_z^2\geq 2\bar{A}_z=: b_0
\]
and deduce that the optimal solution for the minimization problem \eqref{eq: min bar B} is obtained at $b=b_0$, that is $\bar B_z=0$. Thus the reduced dynamics, $z^{\rm red}$, for $z$ is given by
\begin{equation}
d z^{\rm red}_t=\bar A_z z^{\rm red}_t\,dt+\bar D_z\, du_t \; ,
\label{redadd}
\end{equation}
where $\bar A_z$ is determined in \eqref{eq: bar Az} and $\bar D_z=\sqrt{-\frac{2}{3} \bar A_z}$. 


\RV{In appendix \ref{app: L2norm}, we consider the minimization problem \eqref{eq: min bar B} where the $L^\infty$-norm is replaced by the $L^2$-norm.
}
\begin{remark}
    We observe that the same result expressed by Eq. \eqref{redadd} can be reached through a simpler route, which we briefly outline below. Consider the reduced dynamics in Eq.~\eqref{eq:bar z}, and assume that \( \bar{B}_z = 0 \), meaning that we seek \textit{a priori} a reduced model driven solely by additive noise. In this case, the dynamics takes the form given in Eq.~\eqref{redadd}, namely
\[
dz^{\rm red}_t = \bar{A}_z z^{\rm red}_t\,dt + \bar{D}_z\,du_t \;,
\]
with coefficients \( \bar{A}_z \) and \( \bar{D}_z \) to be determined. The drift term \( \bar{A}_z \) can be computed from Eq.~\eqref{eq: bar Az}, whereas the noise amplitude \( \bar{D}_z \) is obtained, in the steady state, by solving the standard Lyapunov equation \cite{Pavl}:
\begin{equation}
2\bar{A}_z p_{zz}^\infty + \bar{D}_z^2 = 0 \;,
\label{Lyap}
\end{equation}
\RV{where $p^\infty_{zz}=1/3$ is the stationary second moment in the $z$-variable of the original dynamics. Following the procedure outlined in \cite{CM22,CDM22,colangeli2025hybrid}, we require the stationary second moment of $\bar{z}_t$ to coincide with $p^\infty_{zz}$, which leads to $\bar{D}_z = \sqrt{-\frac{2}{3} \bar{A}_z}$.}
\end{remark}


We emphasize that the reduced dynamics governed by Eq.~\eqref{redadd} preserves the  localization property of the original system dynamics characterized by Eq.~\eqref{solaver2}. This preservation demonstrates that our reduction method retains the essential physical behavior of the system while significantly simplifying its mathematical description.

\subsection{Quantification of errors}
In this section we provide a quantitative error estimates between the original dynamics $z$ and the reduced one using Wasserstein distance.

Let $\mathcal{P}_2(\R^d)$ be the set of probability measures on $\R^d$ having finite second moments. Let $\mu, \nu$ be two probability measures in $\mathcal{P}_2(\R^d)$. The Wasserstein distance, $W_2(\mu,\nu)$, between them is defined by
\begin{equation*}
 W_2(\mu,\nu)   =\Big(\min_\pi\int |x-y|^2d\pi(x,y)\Big)^{\frac{1}{2}},
\end{equation*}
where the infimum is taken over $\pi\in \Pi(\mu,\nu)$, which is the set of all probability measures that have the first and second marginals as $\mu$ and $\nu$ respectively. $W_2(\mu,\nu)$ can be formulated as
\begin{equation}
\label{eq: W2 formulation}
 W_2(\mu,\nu)^2=\min \mathbb{E}|X-Y|^2,   
\end{equation}
where the infimum is taken over random variables $X$ and $Y$ having distributions $\mu$ and $\nu$, respectively. Let $m$ and $m'$ are the means of $\mu$ and $\nu$ respectively. By using the transformation
\begin{align*}
\mathbb{E}(|X-Y|^2)&=\mathbb{E}(|(X-m)-(Y-m')+(m-m')|^2)    
\\&=|m-m'|^2+\mathbb{E}[|(X-m)-(Y-m')|^2]+2(m-m')\mathbb{E}[(X-m)-(Y-m')]
\\&=|m-m'|^2+\mathbb{E}[|(X-m)-(Y-m')|^2],
\end{align*}
we obtain the following lemma:
\begin{lemma}[\cite{hamm2023wasserstein}] 
\label{lem: lemW2}
We have
\[
W_2^2(\mu,\nu)=|m-m'|^2+\min_{X\sim \mu, Y\sim \nu} \mathbb{E}|(X-m)-(Y-m')|^2.
\]
\end{lemma}
Given a probability measure $\mu$ and a random variable $X\sim\mu$. The covariance matrix $\Sigma$ is given by
\[
\Sigma=\mathbb{E}[(X-\mathbb{E}X)(X-\mathbb{E}X)^T]=\int (x-\bar{x})(x-\bar x)^T \,d\mu=\int x x^T\,d\mu - \bar{x}\bar{x}^T
\]
where $\bar x=\int x\,d\mu(x)$ is the mean of $\mu$.

From Lemma \ref{lem: lemW2} and \cite[Main Theorem, Sec. 2]{dowson1982frechet} we obtain the following bounds for the Wasserstein distance between two probability measures in terms of their means and covariance matrices.

\begin{lemma}[\cite{dowson1982frechet}]
\label{lem: W2estimate}
Let $m,m'$ be the means and $\Sigma$ and $\Sigma'$ be the covariance matrices of $\mu$ and $\nu$ respectively. Then
\begin{equation}
\label{eq: 2bounds}
|m-m'|^2+\mathrm{tr}[\Sigma+\Sigma'-2(\Sigma\Sigma')^{1/2}]\leq W_2^2(\mu,\nu)\leq |m-m'|^2+\mathrm{tr}[\Sigma+\Sigma'+2(\Sigma\Sigma')^{1/2}]. 
\end{equation}
\end{lemma}
Note that in one dimension, the bounds \eqref{eq: 2bounds} become
\[
(m-m')^2+\Big[\sqrt{\Sigma}-\sqrt{\Sigma'}\Big]^2\leq W_2^2(\mu,\nu)\leq (m-m')^2+\Big[\sqrt{\Sigma}+\sqrt{\Sigma'}\Big]^2. 
\]
Employing the above Lemma, we obtain the following quantitative lower-bound and upper-bound estimates of the error between the reduced dynamics and the original one via their means and covariances.   
\begin{proposition} Let $\mu(t)$ and $\mu'(t)$ be the distributions of $z(t)$ and $z^{\rm red}(t)$ respectively. Then
\begin{multline*}
\Big[\langle z(t)\rangle-\bar z(t)\Big]^2+\Big[\sqrt{p_{zz}(t)-\langle z(t)\rangle^2}-\sqrt{p_{zz}(t)-\bar z(t)^2}\,\Big]^2 \leq W_2^2(\mu(t),\mu'(t))\\ \leq \Big[\langle z(t)\rangle-\bar z(t)\Big]^2+\Big[\sqrt{p_{zz}(t)-\langle z(t)\rangle^2}+\sqrt{\bar p_{zz}(t)-\bar z(t)^2}\,\Big]^2    
\end{multline*}
\end{proposition}
\begin{proof}
The statement of this proposition is a direct consequence of Lemma \ref{lem: W2estimate} applying to the exact dynamics $z(t)$ and the reduced dynamics $z^{\rm red}(t)$ with the means and covariance given by 
\begin{align*}
&m=m(t)=\mathbb{E}(z(t))=\langle z(t)\rangle,\quad m'=m'(t)=\mathbb{E}(
z^{\rm red}(t))=\bar z(t),\\
& \Sigma=\Sigma(t)=\mathbb{E}(z(t)^2)-(\langle z(t)\rangle)^2=p_{zz}(t)-\langle z(t)\rangle^2,\\
& \Sigma'=\Sigma'(t)=\mathbb{E}(z^{\rm red}(t)^2)-(\langle z^{\rm red}(t)\rangle)^2=\bar p_{zz}(t)-\bar z(t)^2,
\end{align*}
Therefore, we get
\begin{multline*}
\Big[\langle z(t)\rangle-\bar z(t)\Big]^2+\Big[\sqrt{p_{zz}(t)-\langle z(t)\rangle^2}-\sqrt{p_{zz}(t)-\bar z(t)^2}\,\Big]^2 \leq W_2^2(\mu(t),\mu'(t))\\ \leq \Big[\langle z(t)\rangle-\bar z(t)\Big]^2+\Big[\sqrt{p_{zz}(t)-\langle z(t)\rangle^2}+\sqrt{\bar p_{zz}(t)-\bar z(t)^2}\,\Big]^2.  
\end{multline*}
This completes the proof of this proposition.
\end{proof}



\section{Conclusion}
\label{sec: conclusion}


Multivariate Brownian motions provide a fundamental mathematical framework for modeling systems of interacting stochastic variables. They are widely used in areas such as financial mathematics, engineering, and biological modeling, where they capture both individual randomness and complex interactions among components. However, as the dimensionality increases, the analysis of these systems becomes substantially more challenging because of intricate correlation structures and escalating computational demands.

In this work, we introduced two complementary approaches for dimensionality reduction in multivariate geometric Brownian motions. The first is an extension of the hybrid invariant manifold method that systematically incorporates the effects of multiplicative noise while isolating the slow component of the drift. The second is an analytical framework based on examining the higher‑order ordinary differential equations governing the evolution of lower‑order statistical moments. The effectiveness of both approaches was demonstrated through a GBM model inspired by localization effects in a two‑state quantum model, where the reduced model successfully retains the main properties of the full dynamics. The low‑dimensional structure of the quantum model allows explicit computation of all relevant quantities.
In principle, the reduction scheme presented in Section\ref{sec: reduction scheme} can be extended to higher‑dimensional systems. The main technical challenges in such generalizations arise from solving the Lyapunov equation \eqref{eq: red Lyapunov cond} and the invariance equation associated with the invariant manifold method. Lyapunov equations in high‑dimensional settings have been extensively studied both analytically and numerically; see, for example, \cite{kuvcera1973review,kirsten2020order}. For the invariance equation, approximate solutions can be constructed using the Chapman–Enskog method, which provides asymptotic expansions to arbitrary order in the relevant small parameter \cite{Gor2018,CDM22}. 

Future work will focus on extending the proposed techniques to high‑dimensional systems of physical relevance, as well as to nonlinear extensions of the two‑state model \cite{blanchard2000effective,blanchard2000classical} and GBMs with nonlinear drifts \cite{giordano2023infinite}.

\appendix


\section{Two-variable formulation of the two-state quantum system}
\label{app:appA}

The system of linear SDEs \eqref{meq}, which describes a two-state quantum system, serves as the main theoretical testbed for applying our model reduction procedure. Here we demonstrate how these equations can be reduced to a system of two stochastic equations using two different approaches, one employing an alternative coordinate system (with the trade-off being a loss of linearity in the resulting formulation) and the other based on the conservation of the spin.
Given that the process is constrained to evolve on a unit sphere, we can reformulate the stochastic equations \eqref{meq} adopting spherical coordinates:
\begin{equation}
\begin{aligned}
x& = \cos \! \theta \cos \! \phi,\\    
y& = \cos \! \theta \sin \! \phi,\\
z& =  \sin \! \theta,
\end{aligned} 
\label{spherical}
\end{equation}
where $\theta \in [- \frac \pi 2,  \frac \pi 2 ]$ and  $\phi \in [0,  2 \pi )$. The above definitions
(\ref{spherical}) can be inverted as
\begin{equation}
\begin{aligned}
\theta& = \arcsin \! \left( z\right),\\
\phi& = \arctan \! \left( \frac y x \right).
\end{aligned} 
\label{spherical2}
\end{equation}
From the first definition in (\ref{spherical2}) and 
from the third equation in (\ref{meq}) 
we have respectively the first and the second equality below
\begin{equation}
d \theta =  \frac{dz}{\sqrt{1-z^2}} =\frac{2\alpha y\, dt}{\sqrt{1-z^2}}
=2\alpha \sin \! \phi \,dt, 
\label{theta}
\end{equation}
while the third equality is obtained from the last two definitions in (\ref{spherical})
taking into account that $\cos \! \theta$ is always non negative. 

\noindent
From the first stochastic equation in (\ref{meq}) and from the first two definitions 
in (\ref{spherical}) we immediately have
\begin{equation}
dx=-2\beta^2 \cos \! \theta \cos \! \phi \,dt-2\beta \cos \! \theta \sin \! \phi \, dW_t,
\label{dxphitheta}
\end{equation}
moreover, from the first definition in (\ref{spherical}) we also have
the first equation below (the term $d \! \cos \! \phi \, d \! \cos \! \theta$ vanishes because 
the variable $\theta$, according to (\ref{theta}), evolves deterministically)
and the second equation below is obtained from  (\ref{theta}) 
taking into account that $d \! \cos \! \theta = - \sin \! \theta d  \theta$:
\begin{equation}
\begin{aligned}
d  x = \cos \! \phi \, d \! \cos \! \theta + \cos \! \theta \, d \! \cos \! \phi =
- 2\alpha \sin \! \theta   \sin \! \phi \cos \! \phi \,dt + \cos \! \theta \, d \! \cos \! \phi,
\end{aligned}
\label{dxphitheta2}
\end{equation}
Since both terms at the right of equations (\ref{dxphitheta}) and (\ref{dxphitheta2})
equal $dx$, we immediately obtain
\begin{equation}
d  \! \cos \! \phi = 
-2\beta^2  \cos \! \phi \,dt-2\beta \sin \! \phi \, dW_t 
+ 2\alpha \tan \! \theta   \sin \! \phi \cos \! \phi \,dt,
\label{dcosphi}
\end{equation}
which, in turn, implies
\begin{equation}
d  \phi =  2\beta  \, dW_t
- 2\alpha \tan \! \theta  \cos \! \phi \,dt,
\label{phi}
\end{equation}
in fact, using It${\hat {\rm o}}$'s calculus one has 
$ \,d  \! \cos \! \phi = \! -\! \sin \! \phi \, d \phi  - \frac 1 2  \! \cos \! \phi (d \phi)^2$
where for $(d \phi)^2$ only terms of order $dt$ are retained,
$i.e.$,  $(d \phi)^2= 4 \beta^2 dt$, so that (\ref{dcosphi}) is an immediate consequence 
of  (\ref{phi}) .

In conclusion, the system of equations (\ref{meq}), \red{which is a reformulation of the original two-state quantum system considered in Section \ref{sec: two state system} using} the three cartesian variables \red{$(x,y,z)$},
is replaced by the following system of equations for the latitude and longitude variables:
\begin{equation}
\begin{aligned}
d &\theta =2\alpha \sin \! \phi \,dt,  \\
d  &\phi =  2\beta \, dW_t - 2\alpha \tan \! \theta  \cos \! \phi \,dt.
\end{aligned} 
\label{final}
\end{equation}
The second of the above equations can be replaced by the following equation for the variable $\sin \phi$:
\begin{equation}
d  \sin \! \phi =  - 2 \beta^2 \sin \! \phi \, dt
+ 2\beta  \cos \! \phi \, dW_t - 2\alpha \tan \! \theta  (\cos \! \phi)^2\,dt.
\label{final2}
\end{equation}
 
An alternative method to reduce Eqs.~\eqref{meq} to a system of two stochastic equations originates from the conservation of the spin.
Namely, using the relation $|x|^2+|y|^2+|z|^2=1$, if we remove only $x$ we obtain
\begin{subequations}
\label{eq: yzSDE}
\begin{align}
dy&=-2\beta^2 y\,dt-2\alpha z \,dt+2\beta \sqrt{1-|y|^2-|z|^2}\, dW_t,\\
dz&=2\alpha y\, dt.
\end{align} 
\end{subequations}


Vice versa, suppose we have \eqref{eq: yzSDE}, and let 
$x=\sqrt{1-y^2-z^2}$. Then by applying Ito's lemma we have
\begin{align*}
    dx&=-\frac{y}{\sqrt{1-y^2-z^2}}\,dy-\frac{z}{\sqrt{1-y^2-z^2}}\,dz+2\beta^2 (1-y^2-z^2)\frac{z^2-1}{(1-y^2-z^2)^{3/2}}\, dt
    \\&=-\frac{y}{x}\Big(-2\beta^2 y\,dt-2\alpha z \,dt+2\beta x\, dW_t\Big)-\frac{z}{x}2\alpha y\, dt-2\beta^2\frac{1-z^2}{x}\,dt
    \\&=-2\beta^2\frac{-y^2+1-z^2}{x}\,dt-2\beta y\, dW_t
    \\&=-2\beta^2 x\, dt-2\beta y\, dW_t.
\end{align*}
Thus, the reduced model \eqref{eq: yzSDE} turns out being equivalent to the original one.
We emphasize that the identity of the Brownian motions governing the $x$ and $y$ components is crucial for the above derivation. 
\section{Model reduction of multivariate geometric Brownian motions}
\label{sec: GBMs}

In this section, we begin by recalling key properties of GBMs. We then present two complementary model reduction approaches: one based on the analysis of higher-order ODEs, and the other employing the IM method. These approaches are developed in parallel to emphasize their respective advantages.



\subsection{Dynamics of the mean and \RV{second moments} of GBMs} 
\label{sec: means-2ndmoments}
\RV{In this section, following~\cite{mao2007stochastic,kloeden2013numerical}, we recall the dynamics of the mean and second moments of the GBM \eqref{generalSDE}, which} play a central role in the analysis of GBMs. We compute the mean
\[
m(t)=\E[X(t)],
\]
and the \RV{second-moment} matrix
\begin{equation}
\label{eq: autocorrelation matrix}
P(t)=\E[X(t)X(t)^T]\in \R^{n\times n}.
\end{equation}
We have
\[
dm(t)=d\E(X(t))=\E[A X\RV{(t)}\,dt +B X\RV{(t)} dW\RV{(t)} +D dU\RV{(t)}]=Am(t)\,dt.
\]
Thus $m(t)$ solves the ODE
\begin{equation}
\label{eq: mean value}
\dot{m}(t)= Am(t), \quad m(0)=m_0=\E(X_0),
\end{equation}
which has the explicit solution
\[
m(t)=e^{At}m_0.
\]
We now derive the evolution equation for the \RV{second-moment} matrix $P(t)$. Applying \^{I}to's product rule yields
\begin{align*}
d P(t)&=d\, \E[X(t)X(t)^T]
\\&=\E\Big[X\RV{(t)}dX\RV{(t)}^T+ dX\RV{(t)} X\RV{(t)}^T + dX\RV{(t)} dX\RV{(t)}^T\Big]
\\&=\E\Big[X\RV{(t)}(X\RV{(t)}^T A^T\, dt+X\RV{(t)}^T B^T dW\RV{(t)}^T+D^T dU\RV{(t)}^T)+(AX\RV{(t)}\,dt+BX\RV{(t)} dW\RV{(t)}+D dU\RV{(t)}) X\RV{(t)}^T\\
&\qquad+ (BX\RV{(t)}X\RV{(t)}^TB^T+ DD^T) dt\Big]
\\&= \Big[P(t) A^T + A P(t)+B P(t) B^T+ DD^T]\,dt.
\end{align*}
Hence, $P(t)$ satisfies the matrix differential equation
\begin{equation}
\label{eq: covariance matrix}
\dot{P}(t)=AP+ P A^T+ B P B^T+DD^T, \quad P(0)=P_0=X_0X_0^T.
\end{equation}
We can reformulate equation \eqref{eq: covariance matrix} as a standard linear system of ODEs. Let $\P\in\R^{n^2}$ be the column vector formed by stacking the entries of the matrix 
$P$ row by row:
\[
\P=(P_{11}, P_{12},\ldots, P_{1n},P_{21},\ldots, P_{2n},\ldots,P_{n1},\ldots, P_{nn})^T\in \mathbb{R}^{n^2}.
\]
For each $k\in \{1,\dots,n^2\}$ define
\[
\P_k=P_{ij},\quad i= \left\lceil \frac{k}{n}\right\rceil,~ j=k-(i-1)n.
\]
Then  Eq. \eqref{eq: covariance matrix} is equivalent to the following linear system of $n^2$ coupled ODEs:
\begin{equation}
    \label{eq: linear ODE}
\dot{\P}=\H\P+D\otimes D^T,
\end{equation}
where $\H\in \R^{n^2\times n^2}$ given by 
\begin{equation}
\label{eq: formula H}
  \H=A\otimes I+I\otimes A+ B\otimes B^T,
\end{equation}
where for $U=(u_{ij}), V=(v_{ij})\in \R^{n\times n}$, $U\otimes V\in\R^{n^2\times n^2}$ denotes their Kronecker product, that is
\[
U\otimes V=\begin{pmatrix}
    u_{11}V& u_{12}V&\ldots& u_{1n}V\\
    u_{21}V&u_{22}V&\ldots& u_{2n}V\\
    \vdots&\vdots&\vdots&\vdots\\
    u_{n1}V& u_{n2}V&\ldots&u_{nn}V
\end{pmatrix}.
\]

\subsection*{Stationary \RV{second-moment} matrix}
Let $P_\infty$ denote the stationary \RV{second-moment} matrix of the general GBM. It follows from equation \eqref{eq: covariance matrix} that $P_\infty$ satisfies the following Lyapunov-type matrix equation
\begin{equation}
\label{eq: stationary covariance matrix}
AP_\infty+P_\infty A^T+B P_\infty B^T+DD^T=0.
\end{equation}
Equivalently, in terms of the vector $\P_\infty=\lim_{t\rightarrow \infty} \P(t)\in \R^{n^2}$, the stationary condition can be written as:
\[
\H \P_\infty+D\otimes D^T=0.
\]
\RV{The existence of solutions to \eqref{eq: stationary covariance matrix} and the conditions on the convergence of $\P(t)$ to $\P_\infty$ as $t\rightarrow \infty$ can be found in \cite{kuvcera1973review}.}

\subsection{Model reduction of systems of linear ODEs}
\label{sec:HamCal}

We derive reduced equations for the mean and \RV{second-moment} dynamics given in \eqref{eq: mean value} and \eqref{eq: covariance matrix}, respectively. Both equations are systems of coupled linear ODEs. We begin by showing that, starting from a general system of coupled linear ODEs, one can derive an exact closed higher-order differential equation for each component of the original system. This result is of independent interest and may prove useful in other contexts.\footnote{While this technique appears to be known in principle, we could not locate a specific reference in the literature.}
\RV{This finding was formally introduced in Sec.~\ref{sec: higher order ODEs} as Proposition~\ref{eq: general derivation}, which we restate and prove below.}

\HamCal*

\begin{proof}
It follows from \eqref{eq: eqnu} that, for all $i\geq 0$, we have 
\begin{equation}
\label{eq: higher derivative}
 u^{(i)}(t)=\frac{d^i u}{dt^i}=F^i u,
\end{equation}
\red{where $F^i$ is the $i$-th power of the matrix $F$.}
Since $p_n$ is the characteristic polynomial of $F$, by Hamilton-Cayley theorem, we have
 \[
 p_n(F)=0.
 \]
 It follows from \eqref{eq: higher derivative} that
 \[
  \sum_{i=0}^n a_i u^{(i)}(t)= \sum_{i=0}^n a_i F^i u= p_n(F) u=0,
 \]
 which is the desired equality \eqref{eq: high derivative eqn}.
\end{proof}
Note that \eqref{eq: high derivative eqn} is an equation for $u(t)\in \R^n$. Thus, each of the components $u_i(t)$, $ i=1,\ldots, n$, of the vector $u(t)$ also satisfies \eqref{eq: high derivative eqn}, which must be equipped with a suitable set of initial conditions \red{obtained by iteratively computing $u^{(i+1)}(0)=Fu^{(i)}(0)$ for $i=0,\ldots, n-2$.}
\begin{remark} In this remark, we provide further comments on Proposition \ref{eq: general derivation}.
\begin{enumerate}
    \item Proposition \ref{eq: general derivation} still holds true if the characteristic polynomial of $F$ is replaced by its minimal polynomial or any annihilating polynomial, since these polynomials also satisfy the Hamilton-Cayley theorem used in Proof of Proposition \ref{eq: general derivation}. By using the minimal polynomial instead of the characteristic one, one may obtain an exact, lower order differential equation for each of the component of $u$.
\item The reverse direction in Proposition \ref{eq: general derivation} also holds true, that is the higher-order differential equation \eqref{eq: high derivative eqn} can be transformed to a system of first order ODEs; however, note that this can be done in many different ways. We observe that two systems of first-order ODEs will give rise to the same higher-order ODE if they share the same characteristic polynomial.
\end{enumerate}
\end{remark}


\subsection{Derivation of the fully reduced SDE}
\label{sec: multipleD}
\RV{In this section, we discuss possible extensions of our reduction scheme discussed in Section \ref{sec: reduction scheme} when the reduced dynamics is $k$-dimensional ($1\leq k<n$) mixed Ornstein–Uhlenbeck (OU) and geometric Brownian process, namely $X^{\rm red}\in \R^k$ satisfies the following SDE:
\begin{equation*}
  d X^{\rm red}(t)= \bar A X^{\rm red}(t)\, dt+\bar B X^{\rm red}(t) dw(t)+\bar D du(t),
\end{equation*}
where $\bar A, \bar B, \bar D\in \R^{k\times k}$ are the reduced drift and diffusion coefficients, $w(t)$ and $u(t)$ are independent Wiener processes of $1$ and $k$ dimension, respectively. In this case, the corresponding reduced model for the deterministic part and \RV{second-moment} equations, respectively, become
    \[
    \dot{\bar{X}}=\bar{A} \bar{X}+\bar{a},\quad\text{and}\quad \dot{\tilde{p}}= \bar{H} \tilde{p}+\bar{h},
    \]
    where $\bar{A}\in \R^{k\times k}, \bar a\in \R^k$, $\bar{H}\in \R^{k^2\times k^2}$, $\bar{h}\in \R^{k^2}$. The diffusion matrices $\bar{B}, \bar{D} \in \R^{k\times k}$ satisfy the Lyapunov matrix equation \begin{equation*} 
\bar A \bar P_\infty+\bar P_\infty \bar A^T+\bar B \bar P_\infty \bar B^T+\bar D \bar D^T=0 \; ,
\end{equation*}  
where $\bar P_\infty\in \R^{k\times k}$ is the sub-matrix of the original invariant matrix $P_\infty$ corresponding to the resolved variables. We find $\bar B$ and $\bar D$ such that the \RV{second-moment} dynamics of the reduced dynamics \eqref{eq: general reduced SDE}, $\bar P$ is closest to the reduced second-moment dynamics obtained in the second step in a suitable norm $\|\cdot\|$ (for instance the $L^\infty$-norm or the $L^2$-norm)
\begin{equation*}
\min_{\bar B,\bar D}\|\bar P-\tilde{p}\|.
\end{equation*}
Both dynamics are inhomogeneous linear ODEs and can be solved using Duhamel’s principle. A multi-dimensional counterpart of Lemma \ref{lem: lem1} is
\[
\sup_{t\geq 0}\|e^{At}-e^{Bt}\|_F^2,
\]
where $\|\cdot\|$ denote the Frobenius norm. Then the equation determining the optimal time $t^*$ becomes
\[
g(t):=\mathrm{tr}\Big[(e^{At}-e^{Bt})^T(Ae^{At}-Be^{Bt})\Big]=0,
\]
which is a nonlinear equation. In some special cases, for instance, when $A$ and $B$ are symmetric and commute the above equation can be simplified significantly by diagonalizing $A$ and $B$. However, in general cases, one needs to invoke to numerical methods such as Newton's method. Additionally, if there exist $\bar B, \bar D$ such that
\[
    \bar{H}= \bar{A}\otimes I+I\otimes\bar{A}+\bar{B}\otimes \bar{B}^T \quad \text{and}\quad \bar h= \bar D\otimes\bar D^T,
    \]
where $I$ is the $k$-dimensional identity matrix, then we can take $\bar P=\tilde{p}$.}
\section{Eigenvalues of $M$}
\label{app:app2}
A straightforward, albeit lengthy, calculation yields the characteristic determinant of $M$ defined in \eqref{eq: 4ODEs}:
\begin{equation}
\label{eq: det M}
p_{\lambda}(M)=|M-\lambda I|=\lambda(\lambda^3 + 10\beta^2\lambda^2 + (16\beta^4 + 16\alpha^2)\lambda + 96\alpha^2\beta^2):=\lambda g(\lambda),
\end{equation}
where we have defined
\[
g(\lambda):=\lambda^3 + 10\beta^2\lambda^2 + (16\beta^4 + 16\alpha^2)\lambda + 96\alpha^2\beta^2.
\]
Thus $M$ always has a zero eigenvalue.  To determine the number of real nonvanishing eigenvalues of $M$ we need to determine the number of real roots of $g$, which is a cubic polynomial with a discriminant given by
\begin{align*}
    \Delta&=1024 \beta^{12} \left[ 9 - 103(\frac{\alpha^2}{\beta^4}) + 4(\frac{\alpha^2}{\beta^4})^2 - 16(\frac{\alpha^2}{\beta^4})^3 \right]
    \\&=1024 \beta^{12}[ 9 - 103r + 4r^2 - 16r^3 ]=:f(r),
\end{align*}
where $r:= \alpha^2/\beta^4$. Direct computation shows that $f'(r) < 0$ for all $r$. Hence $f(r)$ is monotonically decreasing in $r$. Furthermore, it has a real root $r_1\approx0.0876$. Thus as $r\leq r_1$, $f(r)\geq f(r_1)=0$, namely as 
\[
\frac{\alpha^2}{\beta^4}\leq r_1
   \iff |\alpha|\leq\sqrt{r_1}\beta^2,
\]
we have $f(r)\geq0$. It follows that the discriminant of the cubic polynomial $g$ satisfies
\begin{align*}
    \Delta=1024 \beta^{12}f(r)\geq0\quad\text{as}\quad |\alpha|<\sqrt{r}_1 \beta^2.
\end{align*}
Thus under the above relation between $\alpha$ and $\beta$,  $g$ has three real roots, which are all negative since
all its coefficients are positive. Hence $M$ has one zero-eigenvalue and three real negative eigenvalues if and only if $|\alpha|< \sqrt{r_1}\beta^2$. 
\section{Exact closed equation for the mean and second moment of $z$}
\label{app:app3}
\label{sec: exact closed explicit}
\RV{We first derive the exact equations for $m_y$ and $m_z$}. From \eqref{ydet} and \eqref{zdet} one can easily derive the exact closed equation for $m_y$:
\[
\ddot{m}_y=-2\beta^2 \dot{m}_y-2\alpha\dot{m}_z=-2\beta^2 \dot{m}_y-4\alpha^2 m_y,
\]
which thus has the structure \eqref{eq: exact deterministic y and z}. It remains to find the exact closed equation for $m_z$. On the one hand, from \eqref{ydet}, we have
\[
\frac{d}{dt}\Big(e^{2\beta^2 t}m_y\Big)=e^{2\beta^2 t}(\dot{m}_y+2\beta^2 m_y)=-2\alpha e^{2\beta^2 t}\,m_z.
\]
On the other hand, from \eqref{zdet}, we get
\[
\frac{d}{dt}\Big(e^{2\beta^2 t}m_y\Big)=\frac{1}{2\alpha}\frac{d}{dt}\Big(e^{2\beta^2 t}\dot{m}_z\Big)=\frac{1}{2\alpha}e^{2\beta^2 t}(\ddot{m}_z+2\beta^2 \dot{m_z}).
\]
Comparing the two expressions above, we obtain
\[
\ddot{m_z}+2\beta^2 \dot{m_z}=-4\alpha^2 m_z.
\]
which, as expected, also displays the structure of Eq.~\eqref{eq: exact deterministic y and z}.

Now we provide an explicit derivation for the exact closed equation for $p_{zz}$. The advantage of this derivation is that it automatically returns the values of the constant $C$ in \eqref{eq: pzz C} without the need of finding the equilibrium values.

Noting that $\dot{p}_{xx}+\dot{p}_{yy}+\dot{p}_{zz}=0$, thus, $p_{xx}+p_{yy}+p_{zz}$ is preserved, say equal to $1$ (this is the
same conservation as in Section \ref{sec: two state system}). Substituting $p_{xx}=1-p_{yy}-p_{zz}$ into the RHS of the equation for $\dot{p}_{yy}$, \eqref{eq: 4ODEs} reduces to the following system of three ODEs
\begin{subequations}
\label{eq: 3ODEs 1}
\begin{align}
\dot{p}_{yy}&=-4\beta^2 (2p_{yy}+p_{zz}-1)-4\alpha p_{yz}\label{eq1}\\
 \dot{p}_{yz}&=-2\beta^2 p_{yz}-2\alpha p_{zz}+2\alpha p_{yy}\label{eq2}\\
 \dot{p}_{zz}&=4\alpha p_{yz}\label{eq3}.
\end{align}
\end{subequations}
From this system, we will derive the exact, closed equation for $p_{zz}$. From \eqref{eq1} and \eqref{eq3}, we have
\begin{align}
    \frac{d}{dt}\Big[e^{8\beta^2 t}p_{yy}\Big]&=e^{8\beta^2 t}[\dot{p}_{yy}+8\beta^2 p_{yy}]\notag
    \\&=e^{8\beta^2 t}\Big[-4\beta^2 p_{zz}-4\alpha p_{yz}+4\beta^2\Big]\notag
\\&=e^{8\beta^2 t}\Big[-4\beta^2 p_{zz}-\dot{p}_{zz}+4\beta^2\Big].\label{eq4}
\end{align}
From \eqref{eq2} and \eqref{eq3} we have
\begin{align}
    \ddot{p}_{zz}&=4\alpha \dot{p}_{yz}\notag
    \\&=4\alpha[-2\beta^2 p_{yz}-2\alpha p_{zz}+2\alpha p_{yy}]\notag
    \\&=-2\beta^2 \dot{p}_{zz}-8\alpha^2 p_{zz}+8\alpha^2 p_{yy}.\label{eq5}
\end{align}
Combining \eqref{eq4} and \eqref{eq5} yields
\begin{align*}
\frac{d}{dt}\Big[e^{8\beta^2 t}\ddot{p}_{zz}\Big]&=\frac{d}{dt}\Big[e^{8\beta^2 t}(-2\beta^2 \dot{p}_{zz}-8\alpha^2 p_{zz}+8\alpha^2 p_{yy})\Big]
\\&=\frac{d}{dt}\Big[e^{8\beta^2 t}(-2\beta^2 \dot{p}_{zz}-8\alpha^2 p_{zz})\Big]+8\alpha^2 \frac{d}{dt}\Big[e^{8\beta^2 t}p_{yy}\Big]
\\&=e^{8\beta^2 t}\Big[(-2\beta^2\ddot{p}_{zz}-8\alpha^2 \dot{p}_{zz})+8\beta^2(-2\beta^2 \dot{p}_{zz}-8\alpha^2 p_{zz})+8\alpha^2(-4\beta^2 p_{zz}-\dot{p}_{zz}+4\beta^2)\Big].
\end{align*}
Expanding the LHS and canceling out the term $e^{8\beta^2 t}$ we obtain
\begin{equation*}\label{eq:3-order pzz}
\dddot{p}_{zz}+8\beta^2 \ddot{p}_{zz}=(-2\beta^2\ddot{p}_{zz}-8\alpha^2 \dot{p}_{zz})+8\beta^2(-2\beta^2 \dot{p}_{zz}-8\alpha^2 p_{zz})+8\alpha^2(-4\beta^2 p_{zz}-\dot{p}_{zz}+4\beta^2).
\end{equation*}
Finally, rearranging the corresponding terms, we obtain the following exact third-order differential equation for $p_{zz}$
\begin{equation}
\label{eq: exact en for pzz}
\dddot{p}_{zz}=-10\beta^2 \ddot{p}_{zz}-16(\alpha^2+\beta^4)\dot{p}_{zz}-96 \alpha^2\beta^2 p_{zz}+32 \alpha^2 \beta^2.
\end{equation}    
This equation recovers Eq. \eqref{eq: exact equation for variances}.

\section{An algebraic interpretation of the invariance equations}
\label{app:app4}

Let $Q,P,R \in \mathbb{R}^{2 \times 2}$ be given by
\begin{equation*}
Q = \begin{pmatrix}
q_{11} & q_{12} \\
q_{21} & q_{22}
\end{pmatrix}, \quad P = \begin{pmatrix}
0 & a \\
0 & 1
\end{pmatrix}, \quad
R =QP= \begin{pmatrix}
0 & aq_{11} + q_{12} \\
0 & aq_{21} + q_{22}
\end{pmatrix},
\label{defPQ}
\end{equation*}
with $a\in\mathbb{R}$. 
Note that $P$ is idempotent ($P^2=P$).


Define $\xi = a q_{21} + q_{22}=\mathrm{Tr}(R)$. The matrix $R$ has an eigenvalue $\xi$ which coincides with an eigenvalue of $Q$ if and only if $a$ satisfies the (invariance) equation
 \begin{equation}
q_{21} a^2 + (q_{22} - q_{11})a - q_{12} = 0 .
\label{IEgen}
\end{equation}
This is easily shown by noticing that Eq. \eqref{IEgen} is equivalent to 
\begin{equation*}
    \det(Q - \xi I) = \xi^2 - \text{Tr}(Q)\xi + \det(Q) = 0,
    \label{eigen}
\end{equation*}
upon substituting $\xi = a q_{21} + q_{22}$, with $I$ the $2\times 2$ identity matrix. 
Likewise, $\mathbf{u}=(a,1)^T$ is a common eigenvector of $Q$ and $R$, with eigenvalue $\xi$, if and only if \eqref{IEgen} holds. Indeed, formula \eqref{IEgen} coincides with the first component of the eigenvalue equation $R\mathbf{u}=\xi \mathbf{u}$.
Moreover, since $P\mathbf{u}=\mathbf{u}$ and $R=QP$, then $R \mathbf{u}=QP\mathbf{u}=Q\mathbf{u}=\xi \mathbf{u}$, therefore $Q$ and $R$ share the same eigenvector $\mathbf{u}$ associated with the common eigenvalue $\xi$.

Finally, we note that the invariance equation \eqref{IE1} introduced earlier appears as a special case of \eqref{IEgen} when the matrix $Q$ is given by the matrix $Q_{\epsilon}$ in \eqref{eq: deterministicODE}. 

This framework straightofrwardly extends to higher-dimensional systems.

\section{Proofs of Lemmas \ref{lem: lem1} and \ref{lem: lem2}}
\label{app:app6}
In this appendix we provide detailed proofs of Lemmas \ref{lem: lem1}-\ref{lem: lem2}.
\begin{proof}[Proof of Lemma \ref{lem: lem1}]
Without loss of generality, we assume $a<b<0$. Then
\[
f(t):=|e^{at}-e^{bt}|=e^{bt}-e^{at}.
\]
$f$ is a differentiable function of $t$ and, since $a,b<0$ we have
\[
f(0)=0,\lim_{t\rightarrow +\infty}f(t)=0.
\]
\red{Let $t^*\in \mathbb{R}$} solve $f'(t^*)=0$, that is
\[
f'(t^*)=be^{bt^*}-ae^{at^*}=0\quad \Leftrightarrow\quad e^{(b-a)t^*}=\frac{a}{b}.
\]
This gives
\[
t^*=\frac{1}{b-a}\ln\Big(\frac{a}{b}\Big).
\]
\red{Since $(b-a)>0$, when $0\leq t<t^*$, we have  $e^{(b-a)t}<e^{(b-a)t^*}=a/b$, thus $f'(t)>0$ and $f$ is increasing in $[0,t^*]$. On the other hand, when $t> t^*$, then $f'(t)<0$ and $f$ is decreasing in $[t,+\infty)$. This implies that $f$ achieves a maximum at $t=t^*$.}
Thus
\[
\sup_{t\geq 0}|e^{at}-e^{bt}|=f(t^*)=e^{bt^*}-e^{at^*}=\Big(\frac{a}{b}\Big)^\frac{b}{b-a}-\Big(\frac{a}{b}\Big)^\frac{a}{b-a},
\]
which is the claimed statement.
\end{proof}
\begin{proof}[Proof of Lemma \ref{lem: lem2}]
    Let $x:=\frac{a}{b}$. Then since $a<b_0\leq b<0$, we have
    \[
    x\geq \frac{a}{b_0}>1.
    \]
In addition, we have
    \[
    \frac{a}{b-a}=\frac{x}{1-x},\quad \frac{b}{b-a}=\frac{1}{1-x}.
    \]
Using these relations, we can write  $F$ in terms of $x$ as
\[
F(b)=x^{1/x}-x^{x/(1-x)}=:G(x).
\]
Thus the minimization problem becomes
\[
 \min_{a<b_0\leq b<0} F(b)=\min_{x\geq \frac{a}{b_0}>1} G(x).
\]
Since $x>1$, we have
\[
G'(x)=\log(x)\frac{x^{x/(1-x)}}{x-1}>0.
\]
Thus $G$ is increasing and therefore
\[
 \min_{a<b_0\leq b<0} F(b)=\min_{x\geq\frac{a}{b_0}>1}G(x)=G\Big(\frac{a}{b_0}\Big)=F(b_0).
\]
This completes the proof of the lemma.
\end{proof}

\section{Error estimates in $L^2$-norm}
\label{app: L2norm}
In this appendix, we consider another natural norm to measure the distance between two functions, which is the $L^2$-norm:
\[
\|f-g\|_2=\int_0^\infty (f(t)-g(t))^2\, dt.
\]
Using the $L^2$-norm, instead of \eqref{eq: general minimization problem}, we consider the following minimization problem
\begin{equation}
 \label{eq: min prob 2}   
\min_{\bar B_z}\|\tilde{p}_{zz}-\bar p_{zz}\|_2.
\end{equation}
Let $a<b_0\leq b<0$. In stead of the $L^\infty$-norm as in Lemmas \ref{lem: lem1}-\ref{lem: lem2}, now let's us compute the following function
\[
F(b)=\int_0^\infty(e^{at}-e^{bt})^2\,dt
\]
We have
\begin{align*}
F(b) &= \int_{0}^{\infty} \bigl(e^{at}-e^{bt}\bigr)^2\,dt \\
     &= \int_{0}^{\infty} \bigl(e^{2at}-2e^{(a+b)t}+e^{2bt}\bigr)dt \\
     &= \int_{0}^{\infty} e^{2at}dt - 2\int_{0}^{\infty} e^{(a+b)t}dt +\int_{0}^{\infty} e^{2bt}dt \\
     &= -\dfrac{1}{2a}+\frac{2}{a+b}-\dfrac{1}{2b}
\end{align*}
Differentiating $F(b)$ yields 
\begin{align*}
F'(b)= -\frac{2}{(a+b)^2} + \frac{1}{2b^2}=\frac{ (a + b)^2-4b^2}{2b^2 (a + b)^2}=\frac{(a + 3b)(a - b)}{2b^2 (a + b)^2}.
\end{align*}
Since $a<b<0$, we have
\begin{align*}
  (a - b)<0 , \, (a + 3b)<0, \quad\text{and}\quad 2b^2 (a + b)^2>0.
\end{align*}
It follows that
\begin{align*}
F'(b)=\frac{(a + 3b)(a - b)}{2b^2 (a + b)^2}>0.
\end{align*}
Therefore, $F(b)$   is increasing on $[b_0, 0)$, and the minimum is attained at $b = b_0$, namely
\begin{align*}
    \min_{a < b_0 \le b < 0} F(b) = F(b_0) = \frac{2}{a + b_0} - \frac{1}{2a} - \frac{1}{2b_0}.
\end{align*}
Hence, by applying this result to the minimization problem \eqref{eq: min prob 2}, we also obtain the optimal solution is given by $\bar B_z=0$. Thus we obtain the same reduced dynamics \eqref{redadd} for $z^{\rm red}$ \red{as when using the $L^\infty$-norm.}
\section*{Acknowledgments}
This work was carried out under the auspices of the Italian National Group of Mathematical Physics. The research of MC has been developed in the framework of the Research Project National Centre
for HPC, Big Data and Quantum Computing—PNRR Project, funded by the European Union—Next Generation EU. 
The research of HD was supported by an EPSRC Grant (EP/Y008561/1).

\bibliographystyle{abbrv}
\bibliography{refs}

@article{kirsten2020order,
  title={Order reduction methods for solving large-scale differential matrix Riccati equations},
  author={Kirsten, Gerhard and Simoncini, Valeria},
  journal={SIAM Journal on Scientific Computing},
  volume={42},
  number={4},
  pages={A2182--A2205},
  year={2020},
  publisher={SIAM}
}

@article{fenichel1979geometric,
  title={Geometric singular perturbation theory for ordinary differential equations},
  author={Fenichel, Neil},
  journal={Journal of differential equations},
  volume={31},
  number={1},
  pages={53--98},
  year={1979},
  publisher={Academic Press}
}

@article{redmann2021bilinear,
  title={Bilinear Systems---A New Link to $\backslash$mathcalH\_2-norms, Relations to Stochastic Systems, and Further Properties},
  author={Redmann, Martin},
  journal={SIAM Journal on Control and Optimization},
  volume={59},
  number={4},
  pages={2477--2497},
  year={2021},
  publisher={SIAM}
}

@book{benner2017model,
  title={Model reduction and approximation: theory and algorithms},
  author={P. Benner and M. Ohlberger and A. Cohen and K. Willcox},
  year={2017},
  publisher={SIAM}
}

@article{bianconi2023complex,
  title={Complex systems in the spotlight: next steps after the 2021 Nobel Prize in Physics},
  author={G. Bianconi and A. Arenas and J. Biamonte and L. D. Carr and B. Kahng and J. Kertesz and J. Kurths and L. L{\"u} and C. Masoller and A. Motter E and others},
  journal={Journal of Physics: Complexity},
  volume={4},
  number={1},
  pages={010201},
  year={2023},
  publisher={IOP Publishing}
}

@article{lucarini2023theoretical,
  title={Theoretical tools for understanding the climate crisis from Hasselmann’s programme and beyond},
  author={V. Lucarini and M. D. Chekroun},
  journal={Nature Reviews Physics},
  volume={5},
  number={12},
  pages={744--765},
  year={2023},
  publisher={Nature Publishing Group UK London}
}

@article{hasselmann1976stochastic,
  title={Stochastic climate models {Part I. Theory}},
  author={K. Hasselmann},
  journal={Tellus},
  volume={28},
  number={6},
  pages={473--485},
  year={1976},
  publisher={Taylor \& Francis}
}

@book{Bogol,
    author = {N. N. Bogoliubov},
    title = {Problems of Dynamical Theory in Statistical Physics},
    publisher ={Studies in Statistical Mechanics} ,
    year ={1960} 
}

@article{blanchard2000effective,
  title={Effective localization induced by noise and nonlinearity},
  author={Blanchard, Philippe and Pasquini, Michele and Serva, Maurizio},
  journal={Physica D: Nonlinear Phenomena},
  volume={141},
  number={3-4},
  pages={214--220},
  year={2000},
  publisher={Elsevier}
}

@incollection{blanchard2000classical,
  title={Classical Limit: Localization Induced by Noise},
  author={Blanchard, Ph and Pasquini, M and Serva, M},
  booktitle={The Foundations Of Quantum Mechanics},
  pages={63--72},
  year={2000},
  publisher={World Scientific}
}

@article{giordano2023infinite,
  title={Infinite ergodicity in generalized geometric Brownian motions with nonlinear drift},
  author={Giordano, Stefano and Cleri, Fabrizio and Blossey, Ralf},
  journal={Physical Review E},
  volume={107},
  number={4},
  pages={044111},
  year={2023},
  publisher={APS}
}

@article{duong2017variational,
  title={Variational approach to coarse-graining of generalized gradient flows},
  author={M.H. Duong and A. Lamacz and M. A. Peletier and U. Sharma},
  journal={Calc. Var. Partial Differ. Equ.},
  volume={56},
  pages={1--65},
  year={2017},
  publisher={Springer}
}

@article{Vulp09,
year = {2009},
month = {jul},
publisher = {},
volume = {2009},
number = {07},
pages = {P07024},
author = {Villamaina, D and Baldassarri, A and Puglisi, A and Vulpiani, A},
title = {The fluctuation-dissipation relation: how does one compare correlation functions and
responses?},
journal = {Journal of Statistical Mechanics: Theory and Experiment}
}

@article{Hang82,
title = {Stochastic processes: Time evolution, symmetries and linear response},
journal = {Physics Reports},
volume = {88},
number = {4},
pages = {207-319},
year = {1982},
author = {Peter H\"{a}nggi and Harry Thomas}
}

@article{Col12,
year = {2012},
month = {apr},
volume = {2012},
number = {04},
pages = {L04002},
author = {Colangeli, Matteo and Rondoni, Lamberto and Vulpiani, Angelo},
title = {Fluctuation-dissipation relation for chaotic non-Hamiltonian systems},
journal = {Journal of Statistical Mechanics: Theory and Experiment}
}

@article{Bobylev,
	author = {Bobylev, A.  V. },
	journal = {Journal of Statistical Physics},
	number = {2},
	pages = {371--399},
	title = {{Instabilities in the Chapman-Enskog Expansion and Hyperbolic Burnett Equations}},
	volume = {124},
	year = {2006}}

@book{chapman1990mathematical,
  title={The Mathematical Theory of Non-Uniform Gases},
  author={Chapman, Sydney and Cowling, T. G.},
  year={1990},
  publisher={Cambridge University Press}
}

@article{mori1965transport,
  title={Transport, collective motion, and Brownian motion},
  author={Mori, Hazime},
  journal={Progress of Theoretical Physics},
  volume={33},
  number={3},
  pages={423--455},
  year={1965}
}

@article{zwanzig1960ensemble,
  title={Ensemble method in the theory of irreversibility},
  author={Zwanzig, Robert},
  journal={The Journal of Chemical Physics},
  volume={33},
  number={5},
  pages={1338--1341},
  year={1960}
}

@article{chorin2000optimal,
  title={Optimal prediction with memory},
  author={Chorin, Alexandre J and Hald, Ole and Kupferman, Raz},
  journal={Physica D: Nonlinear Phenomena},
  volume={166},
  pages={239--257},
  year={2000}
}

@article{majda2001mathematical,
  title={Mathematical strategies for filtering turbulent dynamical systems},
  author={Majda, Andrew J. and Timofeyev, Ilya and Vanden-Eijnden, Eric},
  journal={Journal of the Atmospheric Sciences},
  volume={58},
  number={6},
  pages={913--930},
  year={2001}
}

@article{snowden2017methodologies,
  title={Methods of model reduction for large-scale biological systems: a survey of current methods and trends},
  author={Snowden, Thomas J. and van der Graaf, Piet H. and Tindall, Marcus J.},
  journal={Bulletin of Mathematical Biology},
  volume={79},
  number={7},
  pages={1449--1486},
  year={2017},
  publisher={Springer},
  doi={10.1007/s11538-017-0277-2}
}

@book{antoulas2005approximation,
  title={Approximation of Large-Scale Dynamical Systems},
  author={Antoulas, Athanasios C.},
  year={2005},
  publisher={SIAM}
}

@article{lam1994dynamic,
  title={Dynamic simplification of reaction mechanisms},
  author={Lam, S. H. and Goussis, D. A.},
  journal={Combustion Science and Technology},
  volume={89},
  number={5-6},
  pages={375--404},
  year={1994}
}

@article{roberts2008model,
  title={Model reduction by invariant manifold theory},
  author={Roberts, A. J.},
  journal={Lecture Notes},
  year={2008},
  note={Available online}
}

@article{rowley2004model,
  title={Model reduction for fluids using balanced proper orthogonal decomposition},
  author={Rowley, Clarence W.},
  journal={International Journal of Bifurcation and Chaos},
  volume={15},
  number={3},
  pages={997--1013},
  year={2005}
}

@article{kevrekidis2003equation,
  title={Equation-free, coarse-grained multiscale computation: enabling microscopic simulators to perform system-level analysis},
  author={Kevrekidis, Ioannis G. and Gear, C. William and Hyman, James M. and others},
  journal={Communications in Mathematical Sciences},
  volume={1},
  number={4},
  pages={715--762},
  year={2003}
}

@book{pavliotis2008multiscale,
  title={Multiscale methods: averaging and homogenization},
  author={Pavliotis, G. and Stuart, A.},
  year={2008},
  publisher={Springer Science \& Business Media}
}

@article{CDM23,
doi = {10.1088/1751-8121/ace948},
year = {2023},
publisher = {IOP Publishing},
volume = {56},
number = {34},
pages = {345003},
author = {M. Colangeli and M. H. Duong and A. Muntean},
title = {Model reduction of Brownian oscillators: quantification of errors and long-time behavior},
journal = {Journal of Physics A: Mathematical and Theoretical}
}

@article{givon2004extracting,
  title={Extracting macroscopic dynamics: model problems and algorithms},
  author={Givon, D. and Kupferman, R. and Stuart, A.},
  journal={Nonlinearity},
  volume={17},
  number={6},
  pages={R55},
  year={2004},
  publisher={IOP Publishing}
}

@article{CM22,
author={M. Colangeli and A. Muntean},   	 
title={Reduced Markovian Descriptions of Brownian Dynamics: Toward an Exact Theory},      	
journal={Frontiers in Physics},      	
volume={10},           	
year={2022},           
DOI={10.3389/fphy.2022.903030},      
ISSN={2296-424X}
}

@article{CDM22,
  title={A reduction scheme for coupled {Brownian} harmonic oscillators },
  author={M. Colangeli  and M. H. Duong and A. Muntean},
  journal={Journal of Physics A: Mathematical and Theoretical},
  volume={55},
  number={},
  pages={505002},
  year={2022},
  publisher={}
}

@book{TonyRoberts,
  title={Model {E}mergent {D}ynamics in {C}omplex {S}ystems},
  author={A.~J. Roberts},
  year={2015},
  publisher={SIAM}
}

@article{Legoll2010,
	Author = {Legoll, F. and Leli{\`e}vre, T.},
	Journal = {Nonlinearity},
	Number = {9},
	Pages = {2131},
	Title = {Effective dynamics using conditional expectations},
	Volume = {23},
	Year = {2010}}

@article{zhang2016effective,
  title={Effective dynamics along given reaction coordinates, and reaction rate theory},
  author={Zhang, W. and Hartmann, C. and Sch{\"u}tte, C.},
  journal={Faraday discussions},
  volume={195},
  pages={365--394},
  year={2016},
  publisher={Royal Society of Chemistry}
}

@article{Checkroun1,
  title={Ruelle-{Pollicott} {Resonances} of {Stochastic} {Systems} in {Reduced} {State} {Space}. {Part} {I}: {Theory}},
  author={Checkroun, M. D. and Tantet, A. and Dijkstra, H. A. and Neelin, J. D.},
  journal={Journal of Statistical Physics},
  volume={179},
  pages={1366--1402},
  year={2020},
  publisher={}
}

@article{Checkroun2,
  title={Ruelle-{Pollicott} {Resonances} of {Stochastic} {Systems} in {Reduced} {State} {Space}. {Part} {II}: {Stochastic} {Hopf} {Bifurcation}},
  author={Tantet, A. and Checkroun, M. D. and  Dijkstra, H. A. and Neelin, J. D.},
  journal={Journal of Statistical Physics},
  volume={179},
  pages={1403--1448},
  year={2020},
  publisher={}
}

@article{hartmann,
  title={Model reduction algorithms for optimal control and importance sampling of diffusions},
  author={Hartmann, C. and Sch{\"u}tte, C. and Zhang, W.},
  journal={Nonlinearity},
  volume={29},
  number={8},
  pages={2298--2326},
  year={2016},
  publisher={{IOP} Publishing}
}

@article{Legoll2019,
	doi = {10.1088/1361-6544/ab34bf},
	url = {https://doi.org/10.1088/1361-6544/ab34bf},
	year = 2019,
	month = {oct},
	publisher = {{IOP} Publishing},
	volume = {32},
	number = {12},
	pages = {4779--4816},
	author = {F. Legoll and T. Leli{\`{e}}vre and U. Sharma},
	title = {Effective dynamics for non-reversible stochastic differential equations: a quantitative study},
	journal = {Nonlinearity}
}

@article{Lelievre2019,
author = {Leli\`{e}vre, T. and Zhang, W.},
title = {Pathwise estimates for effective dynamics: {T}he case of nonlinear vectorial reaction coordinates},
journal = {Multiscale Modeling \& Simulation},
volume = {17},
number = {3},
pages = {1019-1051},
year = {2019}}

@article{Lu2014,
author = {Lu,J.  and Vanden-Eijnden, E.},
title = {Exact dynamical coarse-graining without time-scale separation},
journal = {The Journal of Chemical Physics},
volume = {141},
number = {4},
pages = {044109},
year = {2014}
}

@article{Hartmann2020,
author = {Hartmann, C. and Neureither, L. and Sharma, U.},
title = {Coarse Graining of Nonreversible Stochastic Differential Equations: Quantitative Results and Connections to Averaging},
journal = {SIAM Journal on Mathematical Analysis},
volume = {52},
number = {3},
pages = {2689-2733},
year = {2020}}

@article{Duong2018,
	year = 2018,
	month = {aug},
	publisher = {{IOP} Publishing},
	volume = {31},
	number = {10},
	pages = {4517--4566},
	author = {M. H. Duong and A. Lamacz and M. A. Peletier and A. Schlichting and U. Sharma},
	title = {Quantification of coarse-graining error in {L}angevin and overdamped {L}angevin dynamics},
	journal = {Nonlinearity},
	}

@book{gorban2006model,
  title={Model reduction and coarse-graining approaches for multiscale phenomena},
  author={Gorban, A. N and Kazantzis, N. K. and Kevrekidis, I. G. and {\"O}ttinger, H. C. and Theodoropoulos, C.},
  year={2006},
  publisher={Springer}
}

@book{haken1983synergetics,
  title={Synergetics: An Introduction},
  author={Haken, Hermann},
  year={1983},
  publisher={Springer}
}

@article{ott2005low,
  title={Low dimensional behavior of large systems of globally coupled oscillators},
  author={Ott, Edward and Antonsen, Thomas M.},
  journal={Chaos: An Interdisciplinary Journal of Nonlinear Science},
  volume={18},
  number={3},
  pages={037113},
  year={2008},
  publisher={AIP Publishing},
  doi={10.1063/1.2930766}
}

@article{Wouters2019,
author = {Wouters, J. and Gottwald, G. A.},
title = {Stochastic Model Reduction for Slow-Fast Systems with Moderate Time Scale Separation},
journal = {Multiscale Modeling \& Simulation},
volume = {17},
number = {4},
pages = {1172-1188},
year = {2019}}

@article{Wouters2019b,
author = {Wouters, J.  and Gottwald, G. A. },
title = {Edgeworth expansions for slow-fast systems with finite time-scale separation},
journal = {Proceedings of the Royal Society A: Mathematical, Physical and Engineering Sciences},
volume = {475},
number = {2223},
pages = {20180358},
year = {2019}
}

@book{Risken,
	author = {H.\ Risken},
	publisher = {Springer Verlag, Berlin},
	title = {The {Fokker}-{Planck} Equation},
	year = {1996}}

@book{Zwanzig,
	author = {R.\ Zwanzig},
	publisher = {Oxford University Press},
	title = {Nonequilibrium {S}tatistical {M}echanics},
	year = {2001}}

@article{Ghil,
	author = {M. S.\ Guti\'{e}rrez and V.\ Lucarini and M. D.\ Chekroun and M.\ Ghil},
	journal = {Chaos},
	pages = {053116},
	title = {Reduced-order models for coupled dynamical systems: Data-driven methods and the {K}oopman operator},
	volume = {31},
	year = {2021}}

@book{Pavl,
	author = {G. A.\ Pavliotis},
	publisher = {Springer-Verlag},
	title = {Stochastic Processes and Applications. Diffusion Processes, the {Fokker}-{Planck} and {Langevin} Equations},
	year = {2014}}

@article{Kol,
	author = {A. N.\ Kolmogorov},
	journal = {Dokl. Akad. Nauk SSSR},
	pages = {527--530},
	title = {On conservation of conditionally periodic motions under small perturbations of the {Hamiltonian}},
	volume = {98},
	year = {1954}}

@article{Arn,
	author = {V. I.\ Arnold},
	journal = {Russian Math Surveys},
	pages = {9--36},
	title = {Proof of a theorem of {A}.{N}. {Kolmogorov} on the invariance of quasi-periodic motions under small perturbations of the {Hamiltonian}},
	volume = {18},
	year = {1963}}

@article{Mos,
	author = {J.\ Moser},
	journal = {Math. Ann.},
	pages = {136--176},
	title = {Convergent series expansions for quasi-periodic motions},
	volume = {169},
	year = {1967}}

@book{GorKar05,
	author = {A. N.\ Gorban and I. V.\ Karlin},
	publisher = {Springer-Verlag, Berlin},
	series = {Lect. Notes Phys.},
	title = {Invariant Manifolds for Physical and Chemical Kinetics},
	volume = {660},
	year = {2005}}

@article{Gor2018,
title = {Model reduction in chemical dynamics: slow invariant manifolds, singular perturbations, thermodynamic estimates, and analysis of reaction graph},
journal = {Current Opinion in Chemical Engineering},
volume = {21},
pages = {48-59},
year = {2018},
author = {AN Gorban}
}

@article{duong2025coarse,
  title={Coarse graining of stochastic differential equations: averaging and projection method},
  author={Duong, Manh Hong and Hartmann, Carsten and Ottobre, Michela},
  journal={arXiv preprint arXiv:2506.14939},
  year={2025}
}

@article{hamm2023wasserstein,
  title={{On Wasserstein distances for affine transformations of random vectors}},
  author={Hamm, Keaton and Korzeniowski, Andrzej},
  journal={arXiv preprint arXiv:2310.03945},
  year={2023}
}

@article{dowson1982frechet,
  title={{The Fr{\'e}chet distance between multivariate normal distributions}},
  author={Dowson, DC and Landau, BV666017},
  journal={Journal of multivariate analysis},
  volume={12},
  number={3},
  pages={450--455},
  year={1982},
  publisher={Elsevier}
}

@article{Gin2014,
	author = {Ginoux, Jean-Marc},
	journal = {Qualitative Theory of Dynamical Systems},
	number = {1},
	pages = {19--37},
	title = {The Slow Invariant Manifold of the Lorenz--Krishnamurthy Model},
	volume = {13},
	year = {2014}}

@article{colan08,
	author = {I. V.\ Karlin and M.\ Colangeli and M.\ Kr\"{o}ger},
	journal = {Phys. Rev. Lett.},
	pages = {214503},
	title = {Exact linear hydrodynamics from the {Boltzmann} equation},
	volume = {100},
	year = {2008}}

@article{colan09,
	author = {M.\ Colangeli and M.\ Kr\"{o}ger and H. C.\ \"{O}ttinger},
	journal = {Phys. Rev. E},
	pages = {051202},
	title = {Boltzmann equation and hydrodynamic fluctuations},
	volume = {80},
	year = {2009}}

@article{GorKar,
	author = {A. N.\ Gorban and I. V.\ Karlin},
	journal = {Bulletin of the American Mathematical Society},
	pages = {187--246},
	title = {Hilbert's $6$th problem: {E}xact and approximate manifolds for kinetic equations},
	volume = {51},
	year = {2013}}

@article{Stasenko24,
title = {Preservation of dissipativity in dimensionality reduction},
journal = {Communications in Nonlinear Science and Numerical Simulation},
volume = {142},
pages = {108553},
year = {2025},
issn = {1007-5704},
doi = {https://doi.org/10.1016/j.cnsns.2024.108553},
author = {Sergey V. Stasenko and Alexander N. Kirdin}
}

@article{Maes14,
	author = {Christian Maes},
	journal = {Journal of Statistical Physics},
	number = {3},
	pages = {705--722},
	title = {{On the Second Fluctuation--Dissipation Theorem for Nonequilibrium Baths}},
	volume = {154},
	year = {2014}}

@article{Kar02,
	author = {I. V.\ Karlin and A. N.\ Gorban},
	journal = {Annalen der Physik},
	pages = {783--833},
	title = {Hydrodynamics from {Grad}'s equations: {W}hat can we learn from exact solutions?},
	volume = {11},
	year = {2002}}

@article{Gor04,
title = {Constructive methods of invariant manifolds for kinetic problems},
journal = {Physics Reports},
volume = {396},
number = {4},
pages = {197-403},
year = {2004},
issn = {0370-1573},
author = {A. N. Gorban and I. V. Karlin and A. Yu. Zinovyev}
}

@article{colan07,
	author = {M.\ Colangeli and I. V.\ Karlin and M.\ Kr\"{o}ger},
	journal = {Phys. Rev. E},
	pages = {051204},
	title = {From hyperbolic regularization to exact hydrodynamics for linearized {Grad}'s equations},
	volume = {75},
	year = {2007}}

@article{Kubo66,
	author = {R.\ Kubo},
	journal = {Rep. Prog. Phys.},
	pages = {255},
	title = {The fluctuation-dissipation theorem},
	volume = {29},
	year = {1966}}

@book{Kubo,
	author = {R.\ Kubo and M.\ Toda and N.\ Hashitsume},
	publisher = {Springer-Verlag, Berlin},
	title = {Statistical {Physics II}. {N}onequilibrium {S}tatistical {M}echanics},
	year = {1985}}

@article{Vulp,
	author = {U. B. M.\ Marconi and A.\ Puglisi and L.\ Rondoni and A.\ Vulpiani},
	journal = {Phys. Rep.},
	pages = {111--195},
	title = {Fluctuation-{Dissipation}: Response {Theory} in {Statistical} {Physics}},
	volume = {461},
	year = {2008}}

@article{kuvcera1973review,
  title={A review of the matrix Riccati equation},
  author={Ku{\v{c}}era, Vladim{\'\i}r},
  journal={Kybernetika},
  volume={9},
  number={1},
  pages={42--61},
  year={1973},
  publisher={Institute of Information Theory and Automation AS CR}
}

@inproceedings{damm2014balanced,
  title={Balanced Truncation for Stochastic Linear Systems with Guaranteed Error Bound},
  author={Damm, Tobias and Benner, Peter},
  booktitle={21st International Symposium on Mathematical Theory of Networks and Systems},
  pages={1492--1497},
  year={2014}
}

@article{benner2011lyapunov,
  title={Lyapunov equations, energy functionals, and model order reduction of bilinear and stochastic systems},
  author={Benner, Peter and Damm, Tobias},
  journal={SIAM journal on control and optimization},
  volume={49},
  number={2},
  pages={686--711},
  year={2011},
  publisher={SIAM}
}

@article{benner2015model,
  title={Model reduction for stochastic systems},
  author={Benner, Peter and Redmann, Martin},
  journal={Stochastic Partial Differential Equations: Analysis and Computations},
  volume={3},
  pages={291--338},
  year={2015},
  publisher={Springer}
}

@article{colangeli2025hybrid,
  title={A hybrid approach to model reduction of {Generalized} {Langevin} {Dynamics}},
  author={Colangeli, Matteo and Duong, Manh Hong and Muntean, Adrian},
  journal={Journal of Statistical Physics},
  volume={192},
  number={2},
  pages={1--20},
  year={2025},
  publisher={Springer}
}

@book{kloeden2013numerical,
  title={Numerical Solution of Stochastic Differential Equations},
  author={Kloeden, P.E. and Platen, E.},
  isbn={9783662126165},
  lccn={92015916},
  series={Stochastic Modelling and Applied Probability},
  url={https://books.google.co.uk/books?id=r9r6CAAAQBAJ},
  year={2013},
  publisher={Springer Berlin Heidelberg}
}

@book{mao2007stochastic,
  title={Stochastic differential equations and applications},
  author={Mao, Xuerong},
  year={2007},
  publisher={Elsevier}
}

@article{blanchard1994localization,
  title={Localization stabilized by noise},
  author={Blanchard, Ph and Bolz, G and Cini, M and De Angelis, GF and Serva, M},
  journal={Journal of statistical physics},
  volume={75},
  pages={749--755},
  year={1994},
  publisher={Springer}
}

\end{document}